\pdfoutput=1
\RequirePackage{fix-cm}
\documentclass[12pt, a4paper]{article}          
\usepackage{amsthm}
\usepackage{graphicx}
\usepackage{amsmath, amssymb, setspace, eucal, mathrsfs}
\usepackage{subfigure}
\usepackage{pdflscape}
\usepackage{comment}
\usepackage[nonumberlist]{glossaries}
\usepackage{mathtools}
\usepackage{natbib}
\usepackage[ruled, vlined]{algorithm2e}
\usepackage{hyperref}

\def\lcu{\left\{}
\def\rcu{\right\}}
\def\({\left(}
\def\){\right)}
\def\[{\left[}
\def\]{\right]}
\def\<{\left<}
\def\>{\right>}
\def\lmid{\;\middle\vert\;}

\DeclarePairedDelimiter{\floor}{\lfloor}{\rfloor}

\def\E{\mathbb{E}}
\def\P{\mathbb{P}}

\newcommand{\cP}{{\mathcal{P}}}

\newcommand{\eP}{{\EuScript{P}}}

\newcommand{\up}{\emph{up}}
\newcommand{\down}{\emph{down}}

\newcommand{\UP}{\emph{UP}}
\newcommand{\DOWN}{\emph{DOWN}}

\newcommand{\reliability}{\ell}

\DeclareFontFamily{OT1}{pzc}{}
\DeclareFontShape{OT1}{pzc}{m}{it}{<-> s * [1.10] pzcmi7t}{}
\DeclareMathAlphabet{\mathpzc}{OT1}{pzc}{m}{it}

\newcommand{\pzcG}{{\mathpzc{G}}}
\newcommand{\fixedGraph}{\pzcG}

\newcommand{\pzcE}{{\mathpzc{E}}}
\newcommand{\fixedEdgeSet}{\pzcE}

\newcommand{\pzcV}{{\mathpzc{V}}}
\newcommand{\fixedVertexSet}{\pzcV}

\newcommand{\pzcv}{{\mathpzc{v}}}
\newcommand{\fixedVertex}{\pzcv}

\newcommand{\customVertex}[1]{\mathpzc{#1}}

\newcommand{\valueVertexSetUp}{\mathbf x}
\newcommand{\valueVertexSet}{\mathbf x}
\newcommand{\valueVertex}{x}
	
\newcommand{\lowerbound}[1][]{\underbar{\ensuremath{\mathbf x}}^{#1}}
\newcommand{\scriptStylelowerbound}[1][]{\underbar{\ensuremath{\scriptstyle \mathbf x}}^{#1}}
\newcommand{\lowerBound}[1][]{\underbar{\ensuremath{\mathbf X}}^{#1}}
\newcommand{\upperBound}{\overline{\mathbf X}}
\newcommand{\upperbound}{\overline{\mathbf x}}
\newcommand{\upperBoundExt}{\ensuremath{\mathbf P}}
\newcommand{\lowerBoundExt}{\ensuremath{\mathbf D}}
\newcommand{\lowerboundExt}{\ensuremath{\mathbf d}}
\newcommand{\upperboundExt}{\ensuremath{\mathbf p}}

\newcommand{\possibleLower}{\underline{\mathscr{X}}}

\newcommand{\generateSubset}[3]{{\tt GenerateSubset}\(#1, #2, #3\)}
\newcommand{\generateSubsetSymbol}{{\tt GenerateSubset}}

\newcommand{\Vup}{\mathbf X}
\newcommand{\SIS}{SIS }
\newcommand{\SIR}{SIR }

\newcommand{\abs}[1]{\left| #1 \right|}

\newtheorem{sectionprop}[subsection]{Proposition}
\newtheorem{sectioncorollary}[subsection]{Corollary}

\newglossary[slg]{symbols}{sym}{sbl}{List of Symbols}
\newglossaryentry{graph}{type=symbols,sort=aa_graph,name={\ensuremath{\fixedGraph = \(\fixedVertexSet, \fixedEdgeSet\)}},description={Graph with vertex set $\fixedVertexSet$ and edge set $\fixedEdgeSet$}}
\newglossaryentry{graphsubgraphvertex}{type=symbols,sort=ab_graphsubgraphvertex,name={\ensuremath{\fixedGraph\<\fixedVertexSet'\>}},description={Subgraph of $\fixedGraph$ induced by the vertex set $\fixedVertexSet'$}}
\newglossaryentry{powerset}{type=symbols,sort=ac_powerset,name={\ensuremath{\eP\(\fixedVertexSet\)}},description={Set of subgraphs of $\fixedGraph$ induced by a vertex subset, or the power set of $\fixedVertexSet$}}
\newglossaryentry{x_up}{type=symbols,sort=ad_x_up,name={\ensuremath{\Vup}},description={The set of \up{} vertices of a network}}
\newglossaryentry{x_up_superset}{type=symbols,sort=ae_x_up_sub_super_ext,name={\ensuremath{\lowerBoundExt_r}},description={Vertices of $\fixedGraph$ which are definitely \up{}}}
\newglossaryentry{x_up_subset}{type=symbols,sort=af_x_up_sub_super_ext,name={\ensuremath{\upperBoundExt_r}},description={Vertices of $\fixedGraph$ which are possibly \up{}}}
\newglossaryentry{x_up_sub_super}{type=symbols,sort=ag_x_up_sub_super,name={\ensuremath{\lowerBound_r, \upperBound_r}},description={Random subsets and supersets of $\Vup$}}
\newglossaryentry{d_values}{type=symbols,sort=ah_x_up_values,name={\ensuremath{\mathscr{D}_r}},description={Possible values of $\lowerBoundExt_r$}}
\newglossaryentry{x_up_values}{type=symbols,sort=ai_x_up_values,name={\ensuremath{\possibleLower}},description={Possible values of $\lowerBound_r$}}
\newglossaryentry{f}{type=symbols,sort=aj_f,name={\ensuremath{F}},description={Event that $\Vup$ is connected}}
\newglossaryentry{fr}{type=symbols,sort=ak_f_r,name={\ensuremath{F_r}},description={Event that $\Vup$ can still be connected conditional on $\lowerBoundExt_r$}}
\newglossaryentry{i_r}{type=symbols,sort=al_i_r,name={\ensuremath{I_r\(\lowerboundExt_r\)}},description={Indicator function of $\P\(F \lmid \lowerBoundExt_r = \lowerboundExt_r\) > 0$}}
\newglossaryentry{abs}{type=symbols,sort=am_abs,name={\ensuremath{\abs{Z}}},description={The number of elements of set $Z$}}	
\newglossaryentry{closedball}{type=symbols,sort=an_ballclosed,name={\ensuremath{\overline{B\(x, r\)}}},description={Closed ball of radius $r$ around $x$}}
\newglossaryentry{openball}{type=symbols,sort=ao_ballopen,name={\ensuremath{B\(x, r\)}},description={Open ball of radius $r$ around $x$}}

\makeglossaries

\newlength{\savedtextfloatsep}
\newlength{\savedintextsep}
\SetCommentSty{commentFont}

\begin{document}

\title{Estimating Residual Connectivity for \\ Random Graphs
}


\author{Rohan Shah, Dirk P. Kroese
        \\\small School of Mathematics and Physics, The University of Queensland, Australia
}


%

\maketitle

\begin{abstract}
Computation of the probability that a random graph is connected is a challenging problem, so it is natural to turn to approximations such as Monte Carlo methods. We describe sequential importance resampling and splitting algorithms for the estimation of these probabilities. The importance sampling steps of these algorithms involve identifying vertices that must be present in order for the random graph to be connected, and conditioning on the corresponding events. We provide numerical results demonstrating the effectiveness of the proposed algorithm. 
\end{abstract}

\glsaddall
\printglossary[title=Notation,style=long]

\section{Introduction}
\label{sec:intro}
In its broadest sense, \emph{network reliability} is the study of the performance characteristics of systems that can be modeled by random graphs. The most common application is the study of communication networks \citep{Cancela2009}, but other applications include electricity networks \citep{Chassin2005,Pagani2013} and air transport networks \citep{Zanin2013,Wilkinson2012,Cardillo2013}. Often such systems are highly reliable, and the problem of estimating failure probabilities for these systems is one of estimating a rare-event probability. The most widely studied network reliability model is the \emph{K-terminal network reliability model}, where the system is operational if a specified set of vertices is connected; see \cite{Gertsbakh2011a} for further details.

We consider the problem of estimating the \emph{residual connectedness reliability} \citep{Lin2015}, defined as follows. Let $\fixedGraph = \(\fixedVertexSet, \fixedEdgeSet\)$ be an undirected graph. Vertices of $\fixedGraph$ fail independently with probability $1 - p$, and edges are considered failed if either of the vertices fails. The assumption that all vertices fail with equal probability is a notational convenience; generalization to the case where vertices fail with different probabilities is straightforward. Failed vertices are said to be \down{} and functioning vertices are said to be \up{}. The overall network is said to be \UP{} if the \up{} subgraph induced by the \up{} vertices is connected, and \DOWN{} otherwise. The probability that the system is in the \UP{} state is called the residual connectedness reliability. Important applications include \emph{Radio Broadcast Networks} and \emph{Mobile Ad hoc networks} \citep{Royer1999,Tanenbaum2002}.

In practice exact computation of the residual connectedness reliability (RCR) is found to be difficult. The formal statement is that computation of the RCR is \emph{NP-hard} \citep{Sutner1991}. That is, computing the RCR is at least as difficult as the hardest problem in the computational complexity class NP \citep{Sipser1996}. For NP-hard problems we generally turn to approximations such as Monte Carlo methods. 

Most established Monte Carlo methods for the related \emph{K-terminal network reliability problem}, such as approximate zero-variance importance sampling \citep{Lecuyer2011}, the merge process \citep{Elperin1991} and generalized splitting \citep{Botev2013} cannot easily be adapted to the RCR problem. These methods generally rely on the network structure function being monotonic, which is not the case for RCR. One method that can be adapted is the Recursive Variance Reduction (RVR) method \citep{Cancela2002}. Another, specifically designed for the RCR model, is the conditional Monte Carlo method in \cite{Shah2014}. We propose a novel sequential importance resampling algorithm for the estimation of the RCR, which can be interpreted as a ``splitting'' algorithm combined with importance sampling. 

The \emph{splitting method}, first described in \cite{Kahn1951}, is a Monte Carlo technique for the estimation of probabilities of rare events. If the rare event can be written as the intersection of nested events, then the rare-event probability can be written as a product of conditional probabilities. These probabilities will typically be much larger than the rare-event probability and can therefore be estimated more accurately. The key feature of the splitting method is that sample paths of some random process are replicated or \emph{split} when they hit some ``intermediate level of rareness''. Although the estimators of the conditional probabilities are generally not independent, under the splitting method their product is still an unbiased estimator of the rare-event probability. The splitting method and similar sequential importance resampling algorithms such as stochastic enumeration \citep{Villien1991a,Au2001,Liu2001,DelMoral2006,Rubinstein2010,Vaisman2014} have found numerous applications including network reliability \citep{Cancela2008}, queuing theory \citep{Glasserman1999,Garvels2000}, particle transmission \citep{Kahn1951} and air traffic control \citep{Jacquemart2013}. 

The rest of this paper is organized as follows. Section \ref{sec:residual_reliability} introduces the RCR model and describes two existing Monte Carlo estimation methods for the RCR problem. Section \ref{sec:splitting_rcr} introduces two splitting algorithms for the estimation of the RCR. Section \ref{sec:transfer_matrix_method} gives an overview of the transfer matrix method, which allows the exact computation of the RCR in some cases. Section \ref{sec:numerical_results} describes a simulation study performed to compare the different estimation method and discusses the results. A list of proofs is given in Appendix \ref{sec:omitted_proofs}. 

\section{The Residual Connectedness Reliability Model\label{sec:residual_reliability}}

\subsection{Definition}

Let $\fixedGraph = \(\fixedVertexSet, \fixedEdgeSet\)$ be a connected graph in which the \emph{state} of a vertex $\fixedVertex \in \fixedVertexSet$ is a binary random variable. If the binary variable takes value $1$ then $\fixedVertex$ is functioning correctly; otherwise the vertex has failed. For notational convenience we assume that these random variables are all independent and identically distributed. This means that each vertex is in the \up{} state with probability $p$ and the \down{} state with probability $1 - p$. The subset of $\fixedVertexSet$ containing the \up{} vertices is denoted by $\Vup$. We will identify the set $\Vup$ with the subgraph $\fixedGraph\<\Vup\>$ which has vertex set $\Vup$ and edge set 
\begin{align*}
\lcu \lcu \fixedVertex_1, \fixedVertex_2\rcu \in \fixedEdgeSet \colon \fixedVertex_1, \fixedVertex_2 \in \Vup\rcu.
\end{align*}
Let $\cP\(\fixedVertexSet\)$ be the set of all subgraphs of $\fixedGraph$ induced by a vertex subset. Define $\varphi$ as the binary function on $\cP\(\fixedVertexSet\)$ which takes value $1$ for connected graphs and $0$ otherwise. When we write $\varphi\(\Vup\)$, this will be $1$ if $\fixedGraph\<\Vup\>$ is connected and $0$ otherwise. We are interested in estimating the \emph{residual connectedness reliability}
\begin{align*}
\reliability\(\fixedGraph, p\) &= \P\(\mbox{$\Vup$ is connected}\) = \P\(\varphi\(\Vup\) = 1\).
\end{align*}

It will be convenient to take some enumeration $\fixedVertex_1, \dots, \fixedVertex_{\abs{\fixedVertexSet}}$ of the vertices $\fixedVertexSet$ and use this as a total ordering. This makes it is possible to write statements such as $\fixedVertex_1 < \fixedVertex_2$, and $\min \lcu \fixedVertex_1, \fixedVertex_2\rcu = \fixedVertex_1$. We also use the notation $\[\fixedVertex, \infty\) = \lcu \customVertex{w} \in \fixedVertexSet \lmid \customVertex{w} \geq \fixedVertex\rcu$. Note that the ordering used is completely arbitrary.

\subsection{Existing Methods}

The Recursive Variance Reduction (RVR) method of \cite{Cancela2002} can be adapted to the RCR problem as follows. Let $U$ be the event that all vertices of $\fixedGraph$ are \up{} and let $\ell$ be the probability that the input random graph model is connected. The RVR method begins by writing
\begin{align}
\ell &= \P\(\varphi\(\Vup\) = 1\) = \varphi\(\fixedGraph\)\P\(U\) + \P\(\varphi\(\Vup\) = 1\lmid U^c\)\P\(U^c\)\nonumber\\
&= p^{\abs{\fixedVertexSet}} + \P\(\varphi\(\Vup\) = 1\lmid U^c\)\(1 - p^{\abs{\fixedVertexSet}}\).\label{eq:rvr_eq_1}
\end{align}
The event $U^c$ is the event that at least one vertex fails. Let $V_1$ be the \emph{first} vertex that fails (with respect to the total ordering of vertices). If no vertex fails then we say that $V_1 = \infty$, but $U^c$ occurring is equivalent to the condition $V_1 < \infty$. In this case we have
\begin{align*}
\P\(V_1 = \customVertex{v}_i\lmid V_1 < \infty\) &= \frac{p^{i-1} (1-p)}{1 - p^{\abs{\fixedVertexSet}}}.
\end{align*}
Equation (\ref{eq:rvr_eq_1}) can be rewritten as 
\begin{align*}
\ell &= p^{\abs{\fixedVertexSet}} + \P\(\varphi\(\Vup\) = 1\lmid V_1<\infty\)\(1 - p^{\abs{\fixedVertexSet}}\). 
\end{align*}
The problem now is to estimate $\P\(\varphi\(\Vup\) = 1\lmid V_1<\infty\)$. But conditional on an observed value of $V_1$, we know that vertex $V_1$ is \down{}, and all vertices smaller (with respect to the total ordering of vertices) than $V_1$ are \up{}. We have therefore fixed the states of some of the vertices. So conditional on $V_1$, we have a random graph where the number of vertices with unknown state is smaller than in the original random graph. Define $\ell^{(1)}\(v_1\) = \P\(\varphi\(\Vup\) = 1\lmid V_1 = v_1\)$. Then 
\begin{align}
\ell &= p^{\abs{\fixedVertexSet}} + \E\[\ell^{(1)}\(V_1\) \lmid V_1 < \infty\]\(1 - p^{\abs{\fixedVertexSet}}\)\label{eq:rvr_eq_2}.
\end{align}
The RVR method first simulates $\left.V_1\lmid V_1 < \infty\right.$. If $\ell^{(1)}\(V_1\)$ can be explicitly computed, then an estimator for $\ell$ is 
\begin{align*}
p^{\abs{\fixedVertexSet}} + \ell^{(1)}\(V_1\)\(1 - p^{\abs{\fixedVertexSet}}\). 
\end{align*}
If $\ell^{(1)}\(V_1\)$ cannot be explicitly computed then the second \down{} vertex $V_2$ is simulated and the expansion used in (\ref{eq:rvr_eq_1}) is applied recursively to $\ell^{(1)}\(V_1\)$. See \cite{Cancela2002} for further details. 

Another approach is the simple conditional Monte Carlo method suggested in \cite{Shah2014}. In this method a random order is selected for the vertices. The states of the vertices are then simulated in that order, until the first \up{} vertex is identified, denoted by $\omega$. The connected component for the identified \up{} vertex is then simulated. At this point the states of some vertices are still unknown. The probability that the random graph is connected is the probability that the connected component of $\omega$ is the \emph{only} connected component. This is the probability that all vertices with unknown state are in fact \down{}, which is a power of $1-p$. See \cite{Shah2014} for further details. We will refer to this method as the conditional Monte Carlo method. 

\section{Splitting for Residual Connectedness Reliability\label{sec:splitting_rcr}}

We construct a fixed-length Markov chain, denoted by $\lcu \lowerBoundExt_r \rcu_{r=0}^R$. Increasing values of $r$ correspond to increasing information about the random graph $\Vup$. At the last time stage we have complete information about $\Vup$, so in fact $\Vup$ is equal in distribution to $\lowerBoundExt_R$. In some cases we will be able to determine that $\varphi\(\Vup\) = 0$ before ``time'' $R$. We will use this Markov chain to construct a classical splitting algorithm. Note the use of $r$ for the time variable; this is intended to emphasize that our ``time'' variable corresponds to the radius of a certain disk. As in Section \ref{sec:residual_reliability} we will assume some arbitrary ordering of the vertex set $\fixedVertexSet$.

\subsection{Notation\label{subsec:splitting_rcr_notation}}

We illustrate the ideas in this section using the $5 \times 5$ grid graph, with $R = 3$ and the value of $\Vup$ shown in Figure \ref{fig:graphical_representations_G}. Recall that the value of $\Vup$ will be the last value of the Markov Chain $\lcu \lowerBoundExt_r \rcu_{r=0}^R$. We assume the vertex ordering shown in Figure \ref{fig:vertex_ordering_picture}. 

We first construct a metric on the set $\fixedVertexSet$ of vertices. For vertices $\fixedVertex_i$ and $\fixedVertex_j$, the \emph{distance} $d\(\fixedVertex_i, \fixedVertex_j\)$ is defined as one less than the number of vertices in the shortest path in $\fixedGraph$ from $\fixedVertex_i$ to $\fixedVertex_j$. For example, the shortest path from $\fixedVertex_i$ to $\fixedVertex_i$ is $\lcu \fixedVertex_i\rcu$, so $d\(\fixedVertex_i, \fixedVertex_i\) = 0$. If $\fixedVertex_i$ and $\fixedVertex_j$ are connected by an edge then the shortest path is $\lcu \fixedVertex_i, \fixedVertex_j\rcu$, so that $d\(\fixedVertex_i, \fixedVertex_j\) = 1$. As we assumed that $\fixedGraph$ was connected, such a path always exists. 
\begin{figure}
  \centering
  \subfigure[Representation of $\Vup$ as a subgraph of $\fixedGraph$. Solid dots represent \up{} vertices. \label{fig:graphical_representations_G}]{
    \includegraphics[scale=0.55]{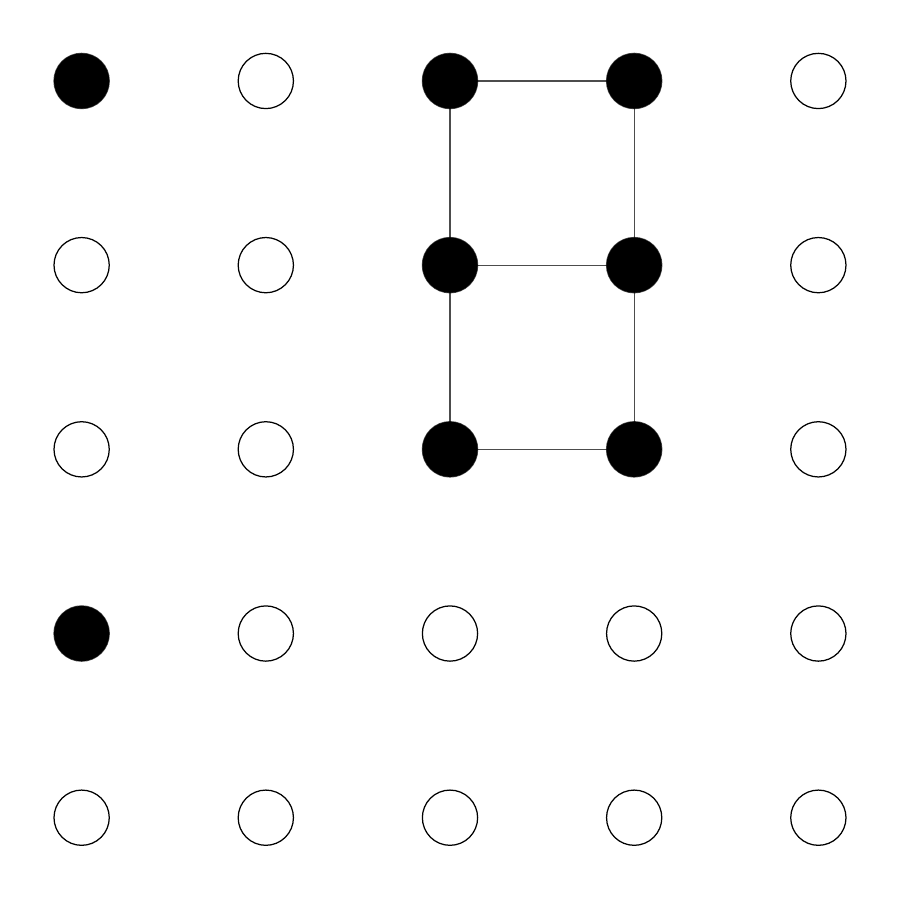}
    }
  \hspace{10pt}
  \subfigure[Ordering of vertices of $\fixedGraph$. \label{fig:vertex_ordering_picture}]{
    \centering
    \includegraphics[scale=0.55]{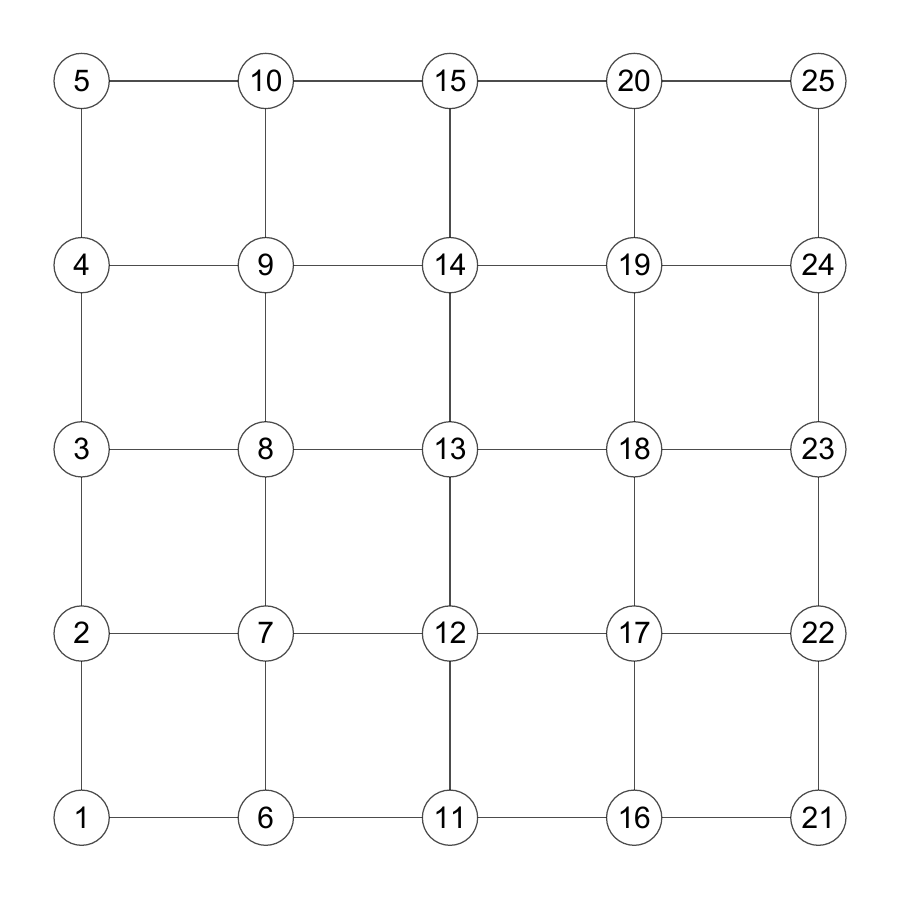}
  }
  \caption{Example value of $\Vup$ (left). Ordering of vertices for the $5 \times 5$ grid graph $\fixedGraph$ (right).}
\end{figure}

Let $R$ be some positive integer. We now define random sets $\lowerBound[n]_r$, for integers $0 \leq r \leq R$ and $n \geq 1$. Recall that $\Vup$ is the set of \up{} vertices, and define $\lowerBound[1]_r = \min \Vup$. That is, $\lowerBound[1]_r$ is the set containing the smallest \up{} vertex with respect to the partial ordering (which does not depend on $r$). For the value of $\Vup$ given in Figure \ref{fig:graphical_representations_G}, $\lowerBound[1]_r = \lcu 2 \rcu$. For a subset $\valueVertexSet$ of vertices, define
\begin{align*}
\mathrm{Next}\(\valueVertexSet, r\) &= \min\lcu \valueVertex \in \Vup \lmid d\(\valueVertex, \valueVertexSet\) > R - r, \valueVertex > \max \valueVertexSet\rcu.
\end{align*}
In words, $\mathrm{Next}\(\valueVertexSet, r\)$ is the next \up{} vertex, after those in $\valueVertexSet$, which is more than distance $R - r$ away from the vertices in $\valueVertexSet$. Such a vertex may not exist, in which case $\mathrm{Next}\(\valueVertexSet, r\) = \emptyset$. We now make the inductive definition
\begin{align*}
\lowerBound[n+1]_r &= \lowerBound[n]_r \cup \mathrm{Next}\(\lowerBound[n]_r, r\). 
\end{align*}
This means that $\lowerBound[n]_r$ is an increasing sequence of sets of \up{} vertices, each set containing at most $n$ vertices. 

Example values of $\lowerBound[1]_0, \lowerBound[2]_0$ and $\lowerBound[3]_0$ are shown in Figures \ref{fig:graphical_representations_partial_lower_bound_1_radius_3} -- \ref{fig:graphical_representations_partial_lower_bound_3_radius_3}. As $\fixedGraph$ is a finite graph, at some point this sequence must reach some largest set, and stop increasing. We therefore define $\lowerBound_r = \bigcup_{n=1}^\infty \lowerBound[n]_r$. Note that if $r = R$ then in fact $\lowerBound_R = \Vup$. In words, $\lowerBound_r$ is the set generated by taking the first \up{} vertex, and then iteratively selecting the \emph{next} (w.r.t. ordering) \up{} vertex that is at least distance $R - r$ away from those previously selected, until no such vertex can be found. Example values of $\lowerBound_0, \lowerBound_1, \lowerBound_2$ and $\lowerBound_3$ are shown in Figures \ref{fig:graphical_representations_partial_lower_bound_3_radius_3} -- \ref{fig:graphical_representations_lower_bound_radius_0}. The set of possible values of $\lowerBound_r$ will be denoted by $\possibleLower_r$. We view the generation of $\lowerBound_r$ as the result of a function $\generateSubset{\Vup}{R}{r}$ defined in Algorithm \ref{alg:subset_generation}. 

\begin{figure}
  \centering
  \subfigure[Value of ${\lowerBound[1]_0}$\label{fig:graphical_representations_partial_lower_bound_1_radius_3}]{
    \centering
    \includegraphics[scale=0.43]{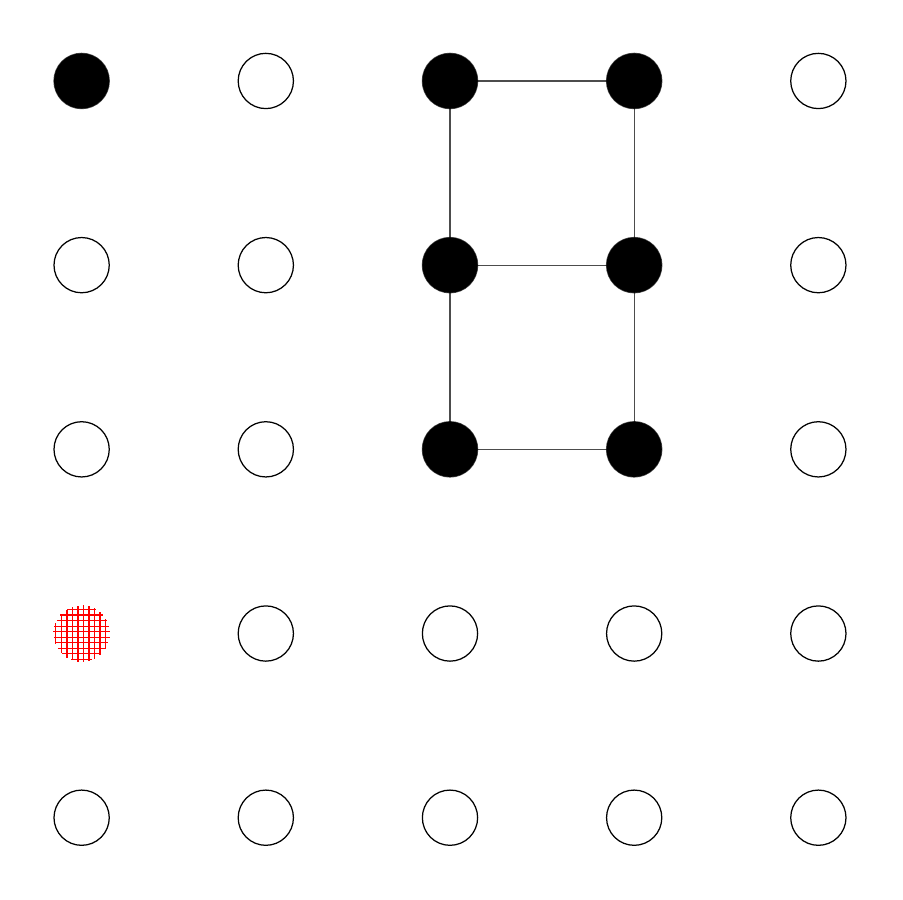}
  }
  \hspace{5pt}
  \subfigure[Value of ${\lowerBound[2]_0 = \lowerBound_0}$\label{fig:graphical_representations_partial_lower_bound_2_radius_3}]{
    \centering
    \includegraphics[scale=0.43]{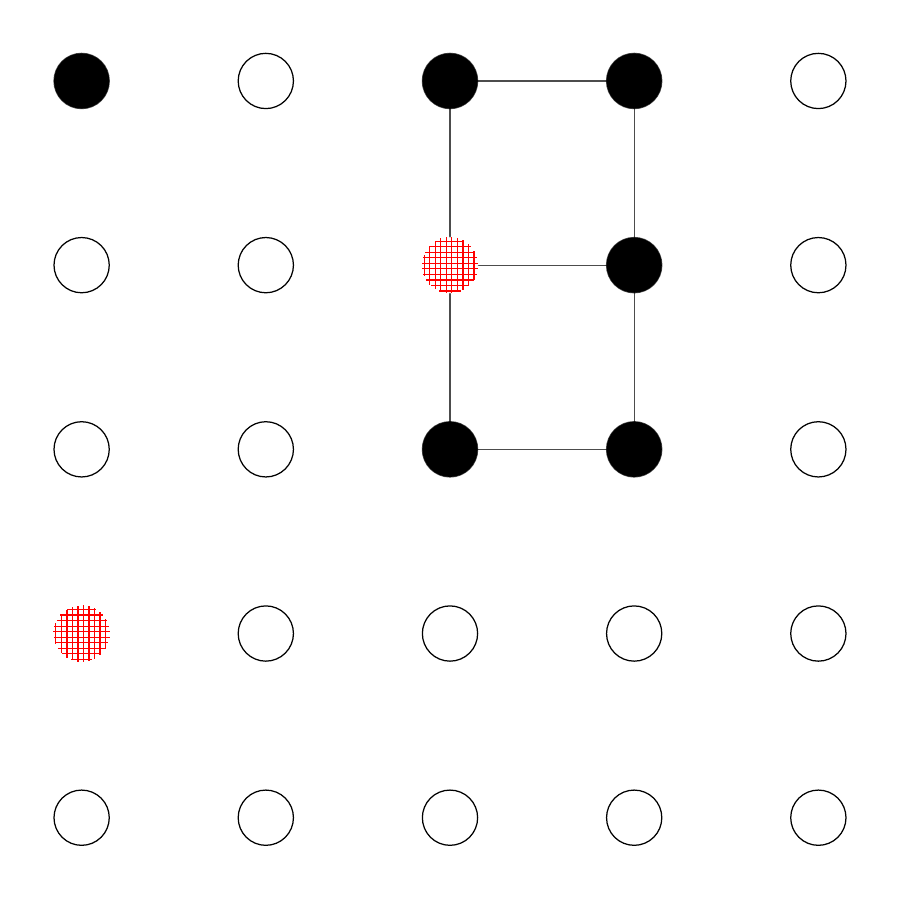}
  }
  \hspace{5pt}
  \subfigure[Value of ${\lowerBound[3]_0 = \lowerBound_0}$\label{fig:graphical_representations_partial_lower_bound_3_radius_3}]{
    \centering
    \includegraphics[scale=0.43]{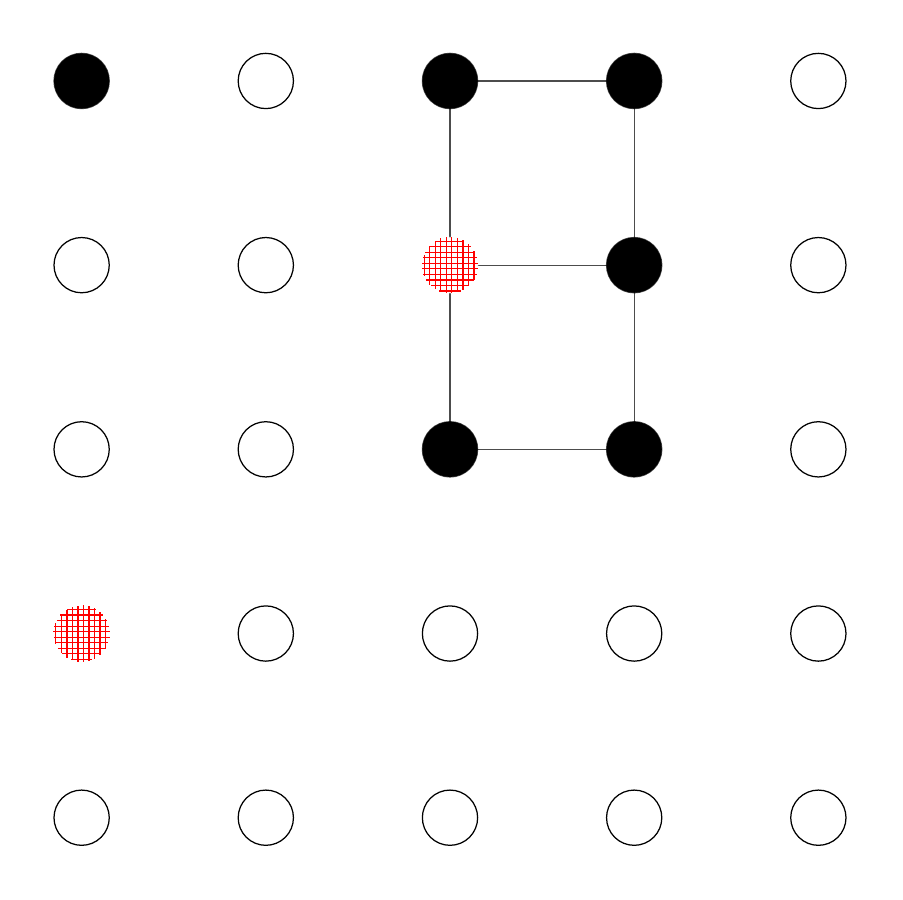}
  }
  \subfigure[Value of $\lowerBound_1$. \label{fig:graphical_representations_lower_bound_radius_2}]{
    \centering
    \includegraphics[scale=0.43]{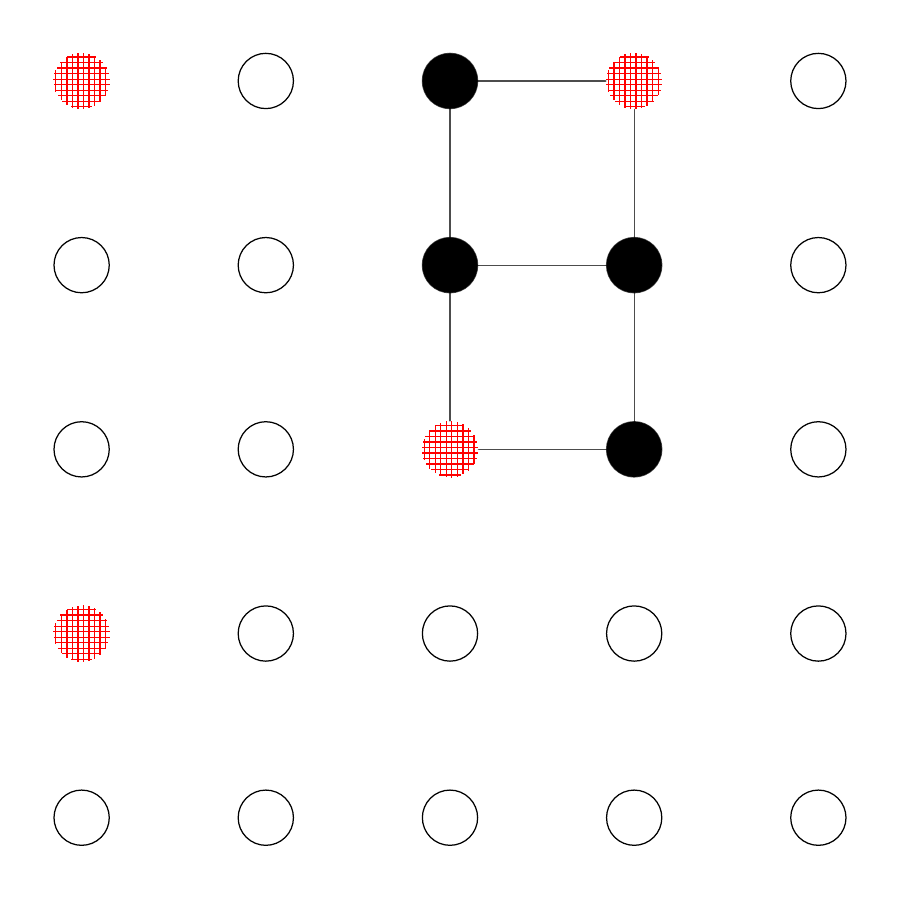}
  }
  \hspace{5pt}
  \subfigure[Value of $\lowerBound_2$. \label{fig:graphical_representations_lower_bound_radius_1}]{
    \centering
    \includegraphics[scale=0.43]{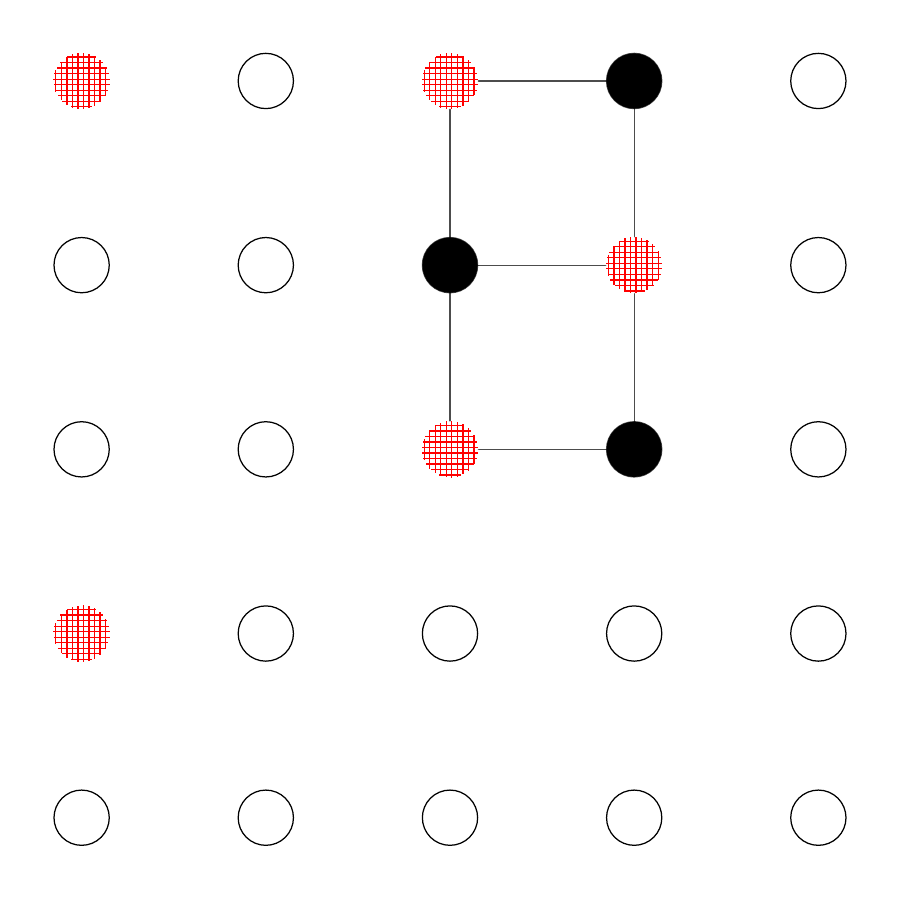}
  }
  \hspace{5pt}
  \subfigure[Value of $\lowerBound_3$. \label{fig:graphical_representations_lower_bound_radius_0}]{
    \centering
    \includegraphics[scale=0.43]{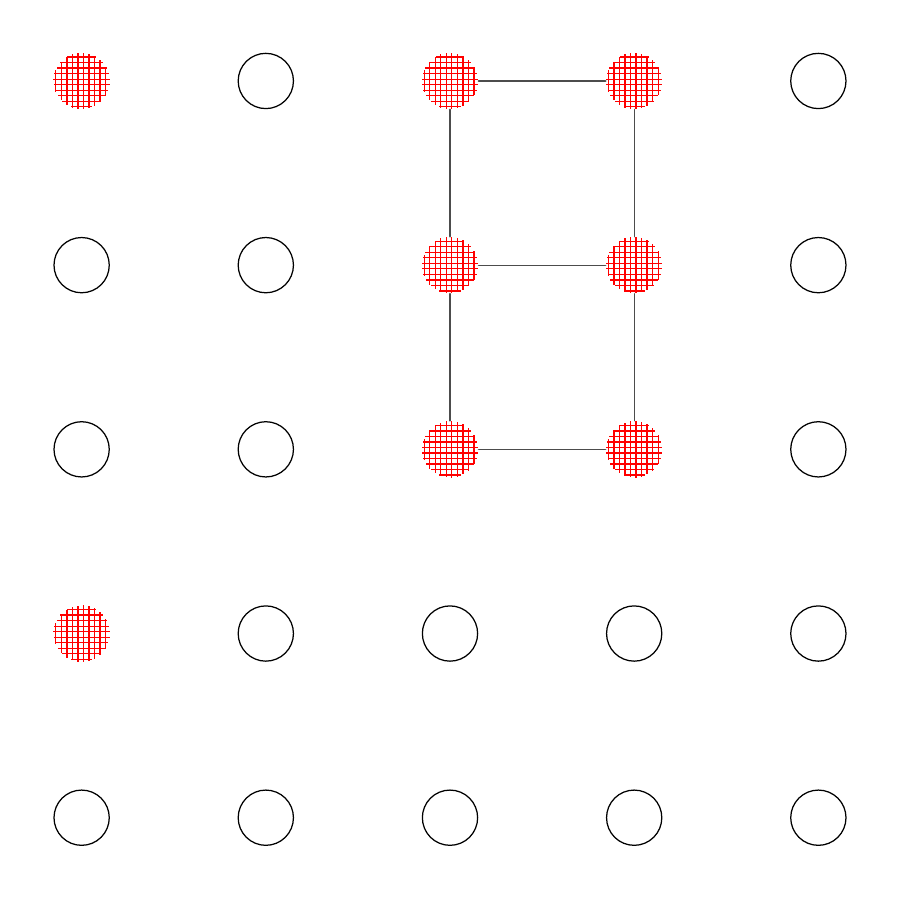}
  }
  \caption{Vertex subsets with $R = 3$ and the value of $\Vup$ shown in Figure \ref{fig:graphical_representations_G}. The red crosshatched vertices are \up{} vertices contained in the specified vertex subset. Solid vertices represent other vertices which will eventually be determined to be \up{}.\label{fig:graphical_representations_partial_lower_bounds}} 
\end{figure}

\LinesNumbered
\begin{algorithm}[!htb]
\SetAlgoSkip{}
\DontPrintSemicolon
\SetKwInOut{Input}{input}\SetKwInOut{Output}{output}
\Input{Set $\Vup \subseteq \fixedVertexSet$, maximum radius $R$, radius $r$}
\Output{$\lowerBound_r \subseteq \fixedVertexSet$}
\If{$\Vup = \emptyset$}{\KwRet{$\emptyset$}}
$\lowerBound_r\leftarrow \emptyset$\;
\While{$\exists s \in \Vup \mbox{ \emph{ such that } } s > \max \lowerBound_r \mbox{ \emph{ and } } d\(s, \lowerBound_r\) > R - r$}{
  $\lowerBound_r \leftarrow \lowerBound_r \cup \min \lcu s \in \Vup \lmid s > \max \lowerBound_r, d\(s, \lowerBound_r\) > R - r\rcu$
}
\KwRet{$\lowerBound_r$}
\caption{$\generateSubset{\Vup}{R}{r}$\label{alg:subset_generation}}
\end{algorithm}
The notation is intended to emphasize that the random sets $\lcu\lowerBound_r\rcu$ are random \emph{subsets} of $\Vup$. The $\lcu \lowerBound_r\rcu$ are all  strongly dependent on $\Vup$ and on each other. The intermediate random sets $\lcu\lowerBound[n]_r\rcu$ are only needed to construct $\lowerBound_r$, and we rarely refer to them past this point. 

For every possible value $\lowerbound_r$ of $\lowerBound_r$ let
\begin{align}
\mathrm{Up}_r\(\lowerbound_r\) &= \bigcup_{\valueVertex \in \scriptStylelowerbound_r} \(\overline{B\(\valueVertex, R - r\)} \cap \left[\valueVertex, \infty\right)\)\label{eq:defn_up_function}. 
\end{align}
The set $\mathrm{Up}_r\(\lowerbound_r\)$ represents all vertices of $\fixedGraph$ that could possibly be \up{} given that $\lowerBound_r = \lowerbound_r$. Define
\begin{align}
\upperBound_r &= \mathrm{Up}_r\(\lowerBound_r\)\label{eq:upper_bound_defn}.
\end{align}
Note that the random sets $\lcu\upperBound_r\rcu$ are random \emph{supersets} of $\Vup$ (see Proposition \ref{prop:upper_bound_proof}). The value of $\lowerBound_r$ completely determines the value of $\upperBound_r$ because of the functional relationship in (\ref{eq:upper_bound_defn}). As $\lowerBound_r$ and $\upperBound_r$ are subsets and supersets of $\Vup$ we have that
\begin{align}
\lowerBound_r \subseteq \Vup \subseteq \upperBound_r, && \mbox{ for all } 0 \leq r \leq R.\label{eq:lower_and_upper_bound}
\end{align}

Finally, we take partial intersections and unions to define the supersets and subsets
\begin{align}
\lowerBoundExt_r &= \bigcup_{t = 0}^r \lowerBound_t, && \mbox{ and } &&  \upperBoundExt_r = \bigcap_{t = 0}^r \upperBound_t.\label{eq:lower_bound_extended_definition}
\end{align}
The mnemonics $D$ and $P$ are chosen because conditional on known $\lowerBound_1, \dots, \lowerBound_r$, the set $\lowerBoundExt_r$ represents the set of vertices that are \emph{definitely} \up{}, and $\upperBoundExt_r$ represents the set of vertices that are \emph{possibly} \up{}. The set of possible values of $\lowerBoundExt_r$ will be denoted by $\mathscr{D}_r$. Example values of $\lowerBoundExt_1$ and $\lowerBoundExt_2$ are shown in Figure \ref{fig:graphical_representations_extended_lower_bounds_1} and \ref{fig:graphical_representations_extended_lower_bounds_2}. Note that the value of $\lowerBoundExt_0$ is always equal to $\lowerBound_0$ which is shown in Figure \ref{fig:graphical_representations_partial_lower_bound_3_radius_3}. The value of $\lowerBoundExt_R$ is always equal to $\Vup$, which is shown in Figure \ref{fig:graphical_representations_G}. 

\begin{figure}
    \centering
    \subfigure[Value of $\lowerBoundExt_1$. \label{fig:graphical_representations_extended_lower_bounds_1}]{
      \centering
      \includegraphics[scale=0.55]{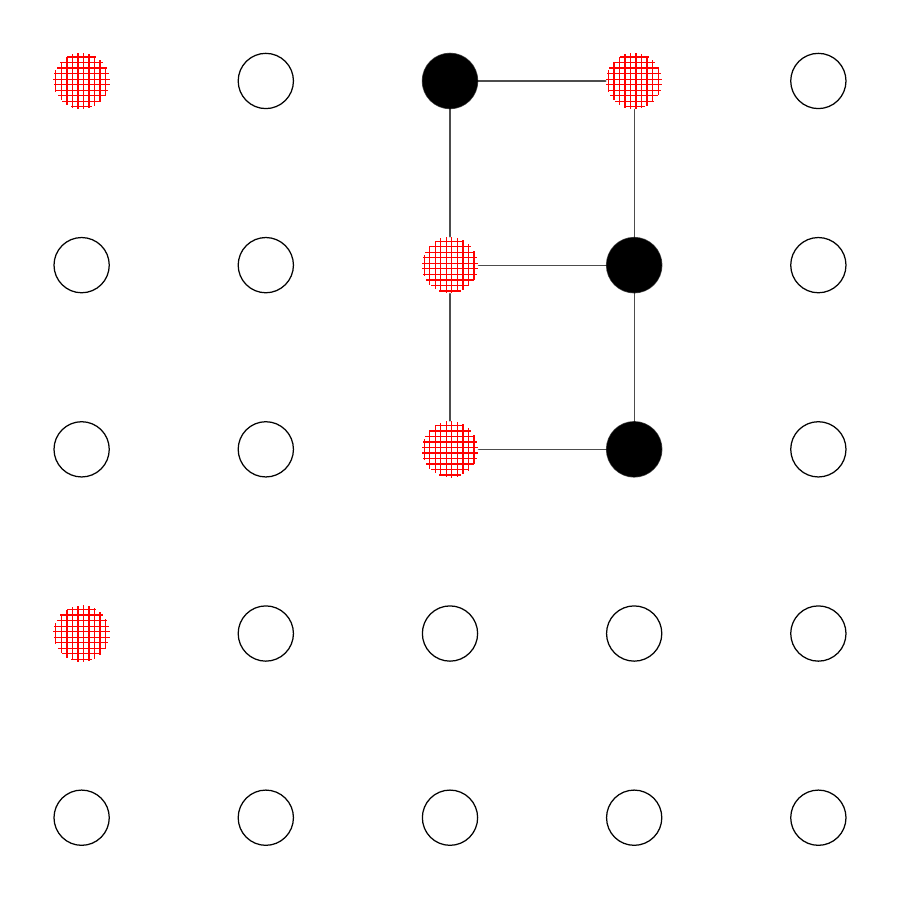}
    }
    \hspace{25pt}
    \subfigure[Value of $\lowerBoundExt_2$.\label{fig:graphical_representations_extended_lower_bounds_2}]{
      \centering
      \includegraphics[scale=0.55]{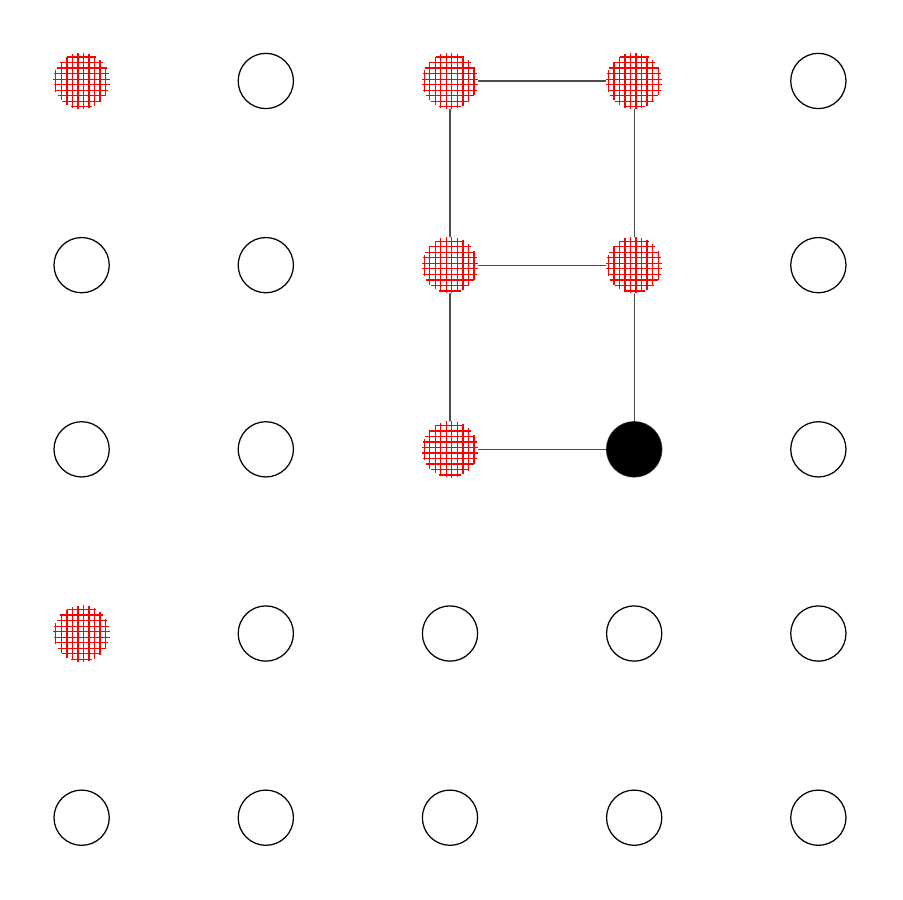}
    }
    \caption{Values of $\lowerBoundExt_1$ and $\lowerBoundExt_2$ with $R = 3$ and the value of $\Vup$ shown in Figure \ref{fig:graphical_representations_G}. The red crosshatched vertices are \up{} vertices contained in the specified vertex subset. Solid vertices represent other vertices which will eventually be determined to be \up{}. \label{fig:graphical_representations_extended_lower_bounds}}
\end{figure}

It is clear from (\ref{eq:lower_and_upper_bound}) that $\lowerBound_r = \lowerbound_r$ implies that $\lowerbound_r \subseteq \Vup \subseteq \mathrm{Up}_r\(\lowerbound_r\)$. However from Proposition \ref{prop:contained_between_determines_lower_bound} the converse is also true. So
\begin{align*}
\P\(\lowerBound_r = \lowerbound_r\) &= \P\(\lowerbound_r \subseteq \Vup \subseteq \mathrm{Up}_r\(\lowerbound_r\)\)
=\P\(\lowerbound_r \subseteq \Vup, \mathrm{Up}_r\(\lowerbound_r\)^c \subseteq \Vup^c\)\\
&= p^{\abs{\lowerbound_r}} \(1 - p\)^{\abs{\mathrm{Up}_r\(\lowerbound_r\)^c}}. 
\end{align*}
From (\ref{eq:lower_and_upper_bound}) and (\ref{eq:lower_bound_extended_definition}) we see that $\lowerBoundExt_r$ and $\upperBoundExt_r$ are subsets and supersets, so that
\begin{align}
\lowerBoundExt_r \subseteq \Vup \subseteq \upperBoundExt_r, && 0 \leq r \leq R.\label{eq:lower_and_upper_bound_ext}
\end{align}
For every $0 \leq s \leq r$ we know that
\begin{align}
\lowerBound_s \subseteq \lowerBoundExt_r \subseteq \Vup \subseteq \upperBoundExt_r \subseteq \mathrm{Up}_s\(\lowerBound_s\). \label{eq:lower_and_upper_bound_ext_2}
\end{align}
Proposition \ref{prop:contained_between_determines_lower_bound} and (\ref{eq:lower_and_upper_bound_ext_2}) imply that knowledge of $\lowerBoundExt_r$ completely determines the value of $\lowerBound_s$ for $0 \leq s \leq r$. Specifically, $\lowerBound_s = \generateSubset{\lowerBoundExt_r}{R}{s}$. Consequently $\lowerBoundExt_r$ also completely determines the values of $\lowerBoundExt_s$ and $\upperBoundExt_s$ for $0 \leq s \leq r$. 

This makes $\lcu \lowerBoundExt_r \rcu_{r=0}^R$ a Markov chain, as its state at any time completely determines its sample path up until that time. By contrast, $\lcu \lowerBound_r \rcu_{r=0}^R$ is \emph{not} a Markov chain. For example, consider Figure \ref{fig:graphical_representations_partial_lower_bound_3_radius_3} and \ref{fig:graphical_representations_lower_bound_radius_2}. Information from $\lowerBound_0$ shows that vertex $14$ of $\fixedGraph$ is \up{}, but this is not reflected in $\lowerBound_1$. So if $X_{14}$ is the binary variable controlling the final state of vertex $14$, we have 
\begin{align*}
\P\(X_{14} = 0\lmid \lowerBound_0 = \lowerbound_0\) &= \P\(X_{14} = 0\lmid \lowerBound_0 = \lowerbound_0, \lowerBound_1 = \lowerbound_1\) = 1,\\
\P\(X_{14} = 0 \lmid \lowerBound_1 = \lowerbound_1\) &\neq 1.
\end{align*}
This means the Markov property fails to hold. 

\subsection{Splitting Algorithm\label{subsec:k_network_splitting_algorithm}}

Let $F$ be the event that $\Vup$ is connected. Let $F_r$ be the event that $\Vup$ can still be connected, conditional on the observed value of $\lowerBoundExt_r$. More precisely, 
\begin{align*}
F &= \lcu \varphi\(\Vup\) = 1 \rcu,\\
F_r &= \lcu \mbox{$\exists \valueVertexSetUp \in \cP\(\fixedVertexSet\)$ with $\lowerBoundExt_r \subseteq \valueVertexSetUp \subseteq \upperBoundExt_r$, such that $\varphi\(\valueVertexSet\) = 1$}\rcu \\
&= \lcu \P\(F \lmid \lowerBoundExt_r\) > 0\rcu. 
\end{align*}
Note that the subsets $\lowerBoundExt_r$ are increasing and the supersets $\upperBoundExt_r$ are decreasing, so that $\lcu F_r \rcu_{r=0}^R$ is a decreasing sequence of events, with $F_R = F$. 

For every $0 \leq r \leq R$ and $\lowerboundExt_{r} \in \mathscr{D}_r$, define $I_r\(\lowerboundExt_{r}\) = \mathbb I \lcu \P\(F \lmid \lowerBoundExt_{r} = \lowerboundExt_{r} \) > 0\rcu$. That is, $I_r$ is the indicator function that takes value $1$ if $\Vup$ can still be connected given the observed value of $\lowerBoundExt_r$, and $0$ otherwise. Note that we can evaluate $I_r\(\lowerBoundExt_{r}\)$ by performing a depth-first search of $\upperBoundExt_r$, started from only \emph{one} of the vertices of $\lowerBoundExt_r$. Let $\valueVertexSet_{\mathrm{DFS}}$ be the vertex set traversed by the depth first search. Then $\valueVertexSet_{\mathrm{DFS}}$ is a connected graph, and if $\valueVertexSet_{\mathrm{DFS}}$ contains $\lowerBoundExt_r$ we have $\lowerBoundExt_r \subseteq \valueVertexSet_{\mathrm{DFS}} \subseteq \upperBoundExt_r$. This makes $\valueVertexSet_{\mathrm{DFS}}$ a connected graph which has a non-zero probability of occurring, in which case 
\begin{align*}
\P\(F \lmid \lowerBoundExt_{r} = \lowerboundExt_{r}\) > \P\(\Vup = \valueVertexSet_{\mathrm{DFS}}\lmid \lowerBoundExt_{r} = \lowerboundExt_{r}\) > 0. 
\end{align*}

The nested events $\lcu F_r \rcu_{r = 0}^R$ allow us to write $\P\(F\)$ as 
\begin{align*}
\P\(F\) &= \P\(F_1 \lmid F_0\)\P\(F_2 \lmid F_1\)\cdots \P\(F_r \lmid F_{r-1}\).
\end{align*}
We can therefore estimate $\P\(F\)$ as a product of estimates of conditional probabilities using a splitting algorithm. This leads to the fixed splitting algorithm given in Algorithm \ref{alg:fixed_splitting_for_network_reliability}. For notational simplicity we take the splitting factors to be integers. Removing this restriction is conceptually easy, see Section 10.6 of \cite{Kroese2011} for further details. \setlength{\intextsep}{10pt}\setlength{\textfloatsep}{10pt}

\begin{algorithm}[!htb]
\SetAlgoSkip{}
\DontPrintSemicolon
\SetKwInOut{Input}{input}\SetKwInOut{Output}{output}
\Input{Connected graph $\fixedGraph$, maximum radius $R$, reliability $p$, integer splitting factors $k_0, \dots, k_{R-1}$, initial number of particles $N$ }
\Output{Estimate of RCR}
$\mathscr{Y}\leftarrow \emptyset$ \tcp*{Samples that have hit an intermediate level}
\For{$i = 1$ \KwTo $N$}
{
  sample $\mathbf D$ from $\lowerBoundExt_0$\;
  \lIf{$I_0\(\mathbf D\)$}
  {
    add $\mathbf D$ to $\mathscr{Y}$
  }
}
\For{$r = 1$ \KwTo $R$}
{
  $\mathscr{Z} \leftarrow \emptyset$ \tcp*{Samples that have hit the next level}
  \For{$\mathbf d_{r-1} \in \mathscr{Y}$}
  {
    \For(\tcp*[f]{Split each particle $k_{r-1}$ times}){$i = 1$ \KwTo $k_{r-1}$} 
    {
      sample $\mathbf D$ from $\(\lowerBoundExt_{r} \lmid \lowerBoundExt_{r-1} = \mathbf d_{r-1}\)$\;
      \lIf{$I_r\(\mathbf D\)$}
      {
        add $\mathbf D$ to $\mathscr{Z}$
      }
    }
  }
  swap $\mathscr{Y}$ and $\mathscr{Z}$
}
\KwRet{$\abs{\mathscr{Y}}\(N\prod_{r=0}^{R-1} k_{r}\)^{-1}$}
\caption{Splitting algorithm for RCR\label{alg:fixed_splitting_for_network_reliability}}
\end{algorithm}

\subsection{Importance Sampling and Resampling Algorithms}\label{sec:importance_resampling_for_rcr}\setlength{\intextsep}{\savedintextsep}\setlength{\textfloatsep}{\savedtextfloatsep}

Experimentally we observe that most particles reach event $F_{R-1}$, but then tend not to reach event $F_R = F$. That is, we have decomposed the rare-event probability into a product of $R$ conditional probabilities, where the last conditional probability is very small and the others are relatively large. Typically in splitting we would solve this by adding events between $F_{R-1}$ and $F_R$. But due to the discrete nature of the constructed Markov chain it is not obvious how this can be done. Fortunately, the particles which reach $F_{R-1}$ have a very special structure. 

In graph theory a \emph{cut vertex} or \emph{articulation vertex} is a vertex of a graph which, when removed, increases the number of connected components. Now, consider the set $V_{\mathrm{cut}}$ of cut vertices of $\upperBound_{R-1}$. If $F_{R-1}$ occurs and some $\fixedVertex \in V_{\mathrm{cut}}$ is not contained in $\lowerBound_{R-1}$, then $\fixedVertex$ must be \up{} if $\upperBound_R = \Vup$ is to be connected. For a proof of this statement see Proposition \ref{prop:articulation_vertices_of_second_last}. Note that this only applies at step $R-1$, and not at any other step. Figure \ref{fig:articulation_vertices} gives a graphical representation of this statement. Figure \ref{fig:articulation_vertices_upper} shows $\upperBoundExt_2$, which is the set all of vertices that are possibly in $\Vup$. Figure \ref{fig:articulation_vertices_lower} shows the smaller set $\lowerBoundExt_2$, which is the set of all vertices that are definitely in $\Vup$. Figure \ref{fig:articulation_vertices_of_upper} shows all the cut vertices of $\upperBoundExt_2$, and Figure \ref{fig:articulation_vertices_of_upper_and_not_lower} shows the single cut vertex of $\upperBoundExt_2$ which is not in $\lowerBoundExt_2$. If this vertex is \down{} then $\Vup$ cannot be connected. 

\begin{algorithm}[!htb]
\SetAlgoSkip{}
\DontPrintSemicolon
\SetKwData{Weights}{\ensuremath{\mathscr{W}}}\SetKwData{CutVertices}{cut}
\SetKwInOut{Input}{input}\SetKwInOut{Output}{output}
\Input{Connected graph $\fixedGraph$, maximum radius $R$, reliability $p$, integer splitting factors $k_0, \dots, k_{R-1}$, initial number of particles $N$ }
\Output{Estimate of RCR}
$\mathscr{Y}\leftarrow \emptyset$ \tcp*{Samples that have hit an intermediate level}
\For{$i = 1$ \KwTo $N$}
{
  sample $\mathbf D$ from $\lowerBoundExt_0$\;
  \lIf{$I_0\(\mathbf D\)$}
  {
    add $\mathbf D$ to $\mathscr{Y}$
  }
}
\For{$r = 1$ \KwTo $R-1$}
{
  $\mathscr{Z} \leftarrow \emptyset$ \tcp*{Samples that have hit the next level}
  \For{$\mathbf d_{r-1} \in \mathscr{Y}$}
  {
    \For(\tcp*[f]{Split each particle $k_{r-1}$ times}){$i = 1$ \KwTo $k_{r-1}$} 
    {
      sample $\mathbf D$ from $\(\lowerBoundExt_{r} \lmid \lowerBoundExt_{r-1} = \mathbf d_{r-1}\)$\;
      \lIf{$I_r\(\mathbf D\)$}
      {
        add $\mathbf D$ to $\mathscr{Z}$
      }
    }
  }
  swap $\mathscr{Y}$ and $\mathscr{Z}$
}
$\Weights\leftarrow \emptyset$ \tcp*{Importance weights}
\For{$\mathbf d_{R-1} \in \mathscr{Y}$}
{
  $\CutVertices \leftarrow $ cut vertices of $\mathrm{Up}_{R-1}\(\mathbf y\)$\;
  \For(\tcp*[f]{Resample each particle $k_{R-1}$ times}){$i = 1$ \KwTo $k_{R-1}$}
  {
    sample $\mathbf D \sim \(\lowerBoundExt_R \lmid \lowerBoundExt_{R-1} = \mathbf d_{R-1}, \CutVertices \subseteq \Vup\)$ \tcp*{Imp. sampling}
    \lIf(\tcp*[f]{Add imp. weight}){$I_{R}\(\mathbf D\)$}
    {
      add $p^{\abs{\CutVertices\setminus \mathbf d_{R-1}}}$ to \Weights
    }
  }
}
\KwRet{$\(N\prod_{r=0}^{R-1} k_{r}\)^{-1}\sum_{w \in \Weights}w$} \tcp*{Return average of imp. weights}
\caption{Sequential importance sampling algorithm for RCR\label{alg:fixed_splitting_and_conditioning_for_network_reliability}}
\end{algorithm}

\begin{figure}
    \centering
    \subfigure[Value of $\upperBoundExt_2$. \label{fig:articulation_vertices_upper}]{
      \centering
      \includegraphics[scale=0.45]{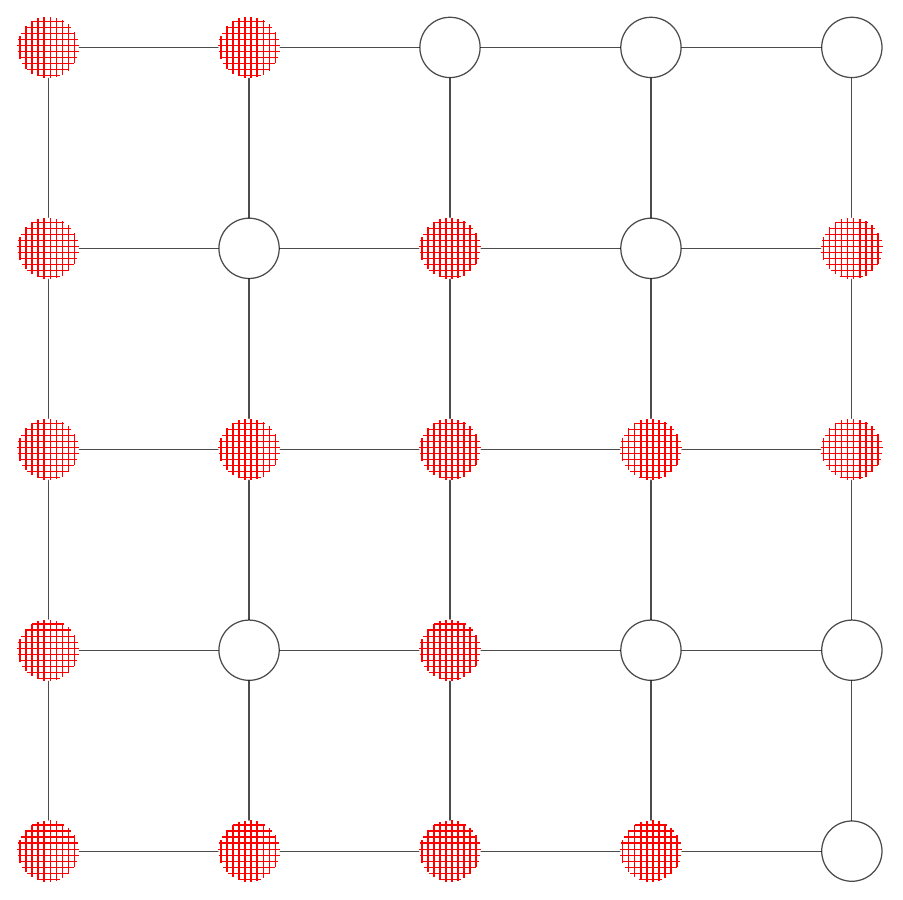}
    }
    \hspace{35pt}
    \subfigure[Value of $\lowerBoundExt_2$. \label{fig:articulation_vertices_lower}]{
      \centering
      \includegraphics[scale=0.45]{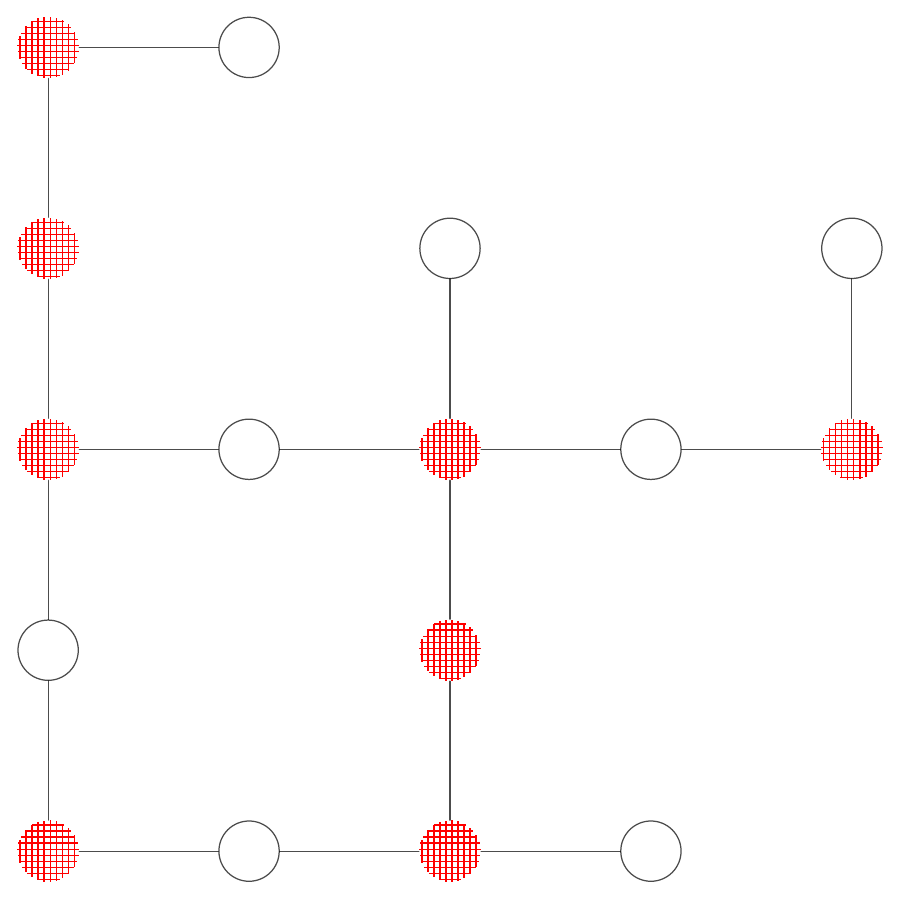}
    }
    \centering
    \subfigure[Cut vertices of $\upperBoundExt_2$. \label{fig:articulation_vertices_of_upper}]{
      \centering
      \includegraphics[scale=0.45]{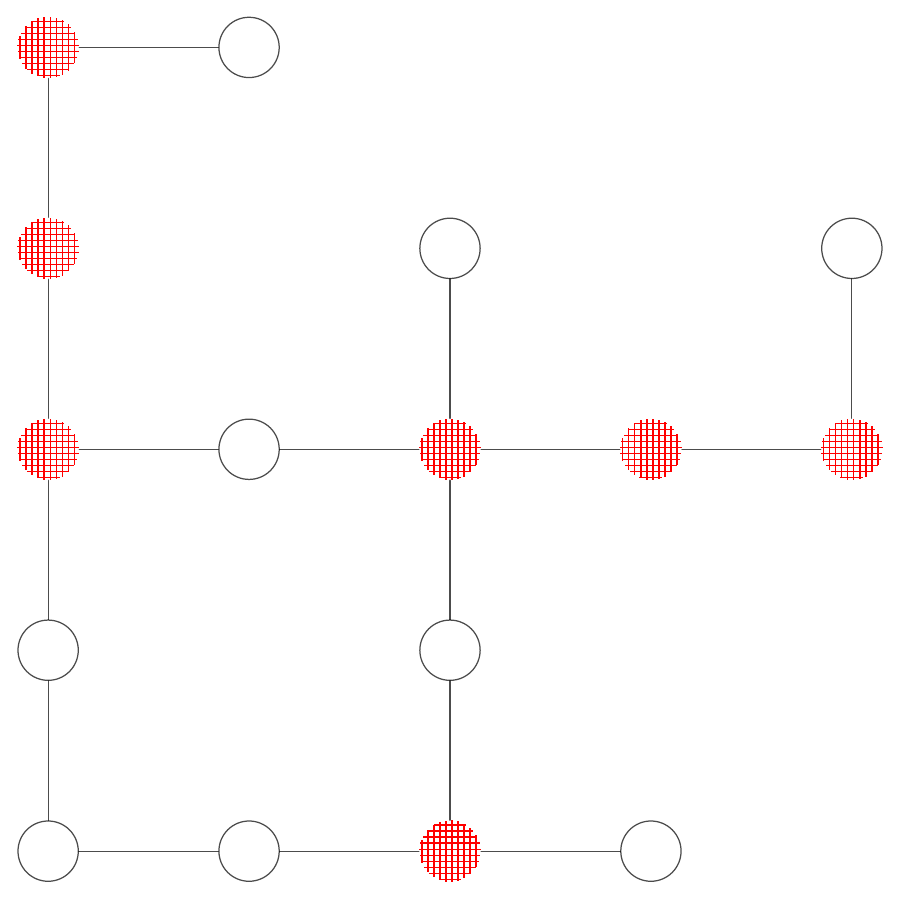}
    }
    \hspace{35pt}
    \subfigure[Cut vertices of $\upperBoundExt_2$ that are not in $\lowerBoundExt_2$. \label{fig:articulation_vertices_of_upper_and_not_lower}]{
      \centering
      \includegraphics[scale=0.45]{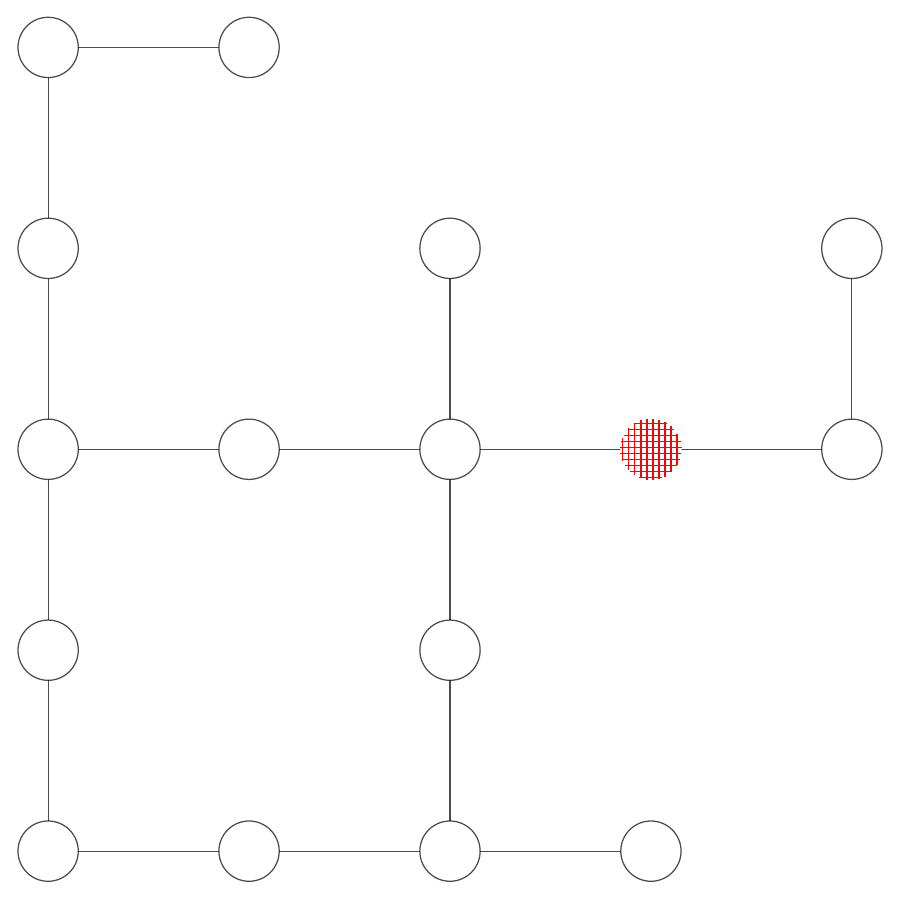}
    }
\caption{Visual representation of the result described in Proposition \ref{prop:articulation_vertices_of_second_last}, with $R = 3$. Figure \ref{fig:articulation_vertices_of_upper_and_not_lower} shows the single cut vertex of $\upperBoundExt_2$ that is not in $\lowerBoundExt_2$. This vertex must be \up{} in order for $\Vup$ to be connected. \label{fig:articulation_vertices}}
\end{figure}
\begin{figure}
\centering
    \subfigure[The single biconnected component of the value of $\upperBoundExt_2$ shown in Figure \ref{fig:articulation_vertices_upper} which contains more than two vertices. \label{fig:biconnected_components_1}]{
      \centering
      \includegraphics[scale=0.45]{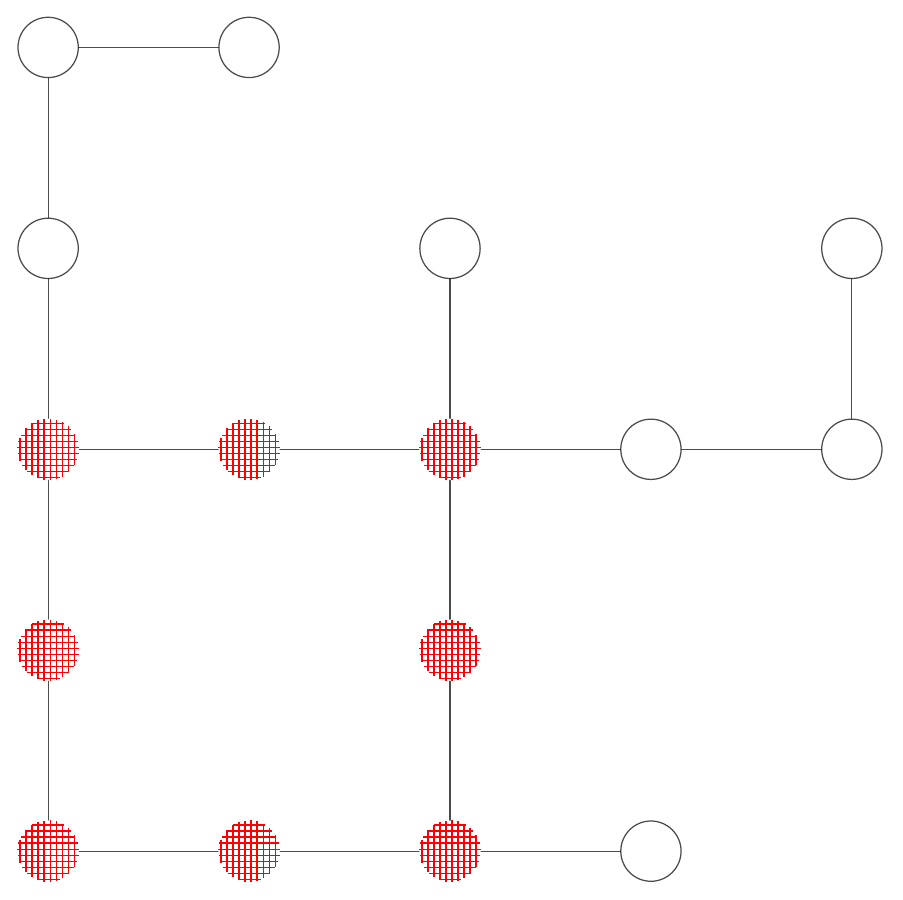}
    }
    \hspace{35pt}
  \subfigure[Another biconnected component of the value of $\upperBoundExt_2$ shown in Figure \ref{fig:articulation_vertices_upper}. \label{fig:biconnected_components_2}]{
      \centering
      \includegraphics[scale=0.45]{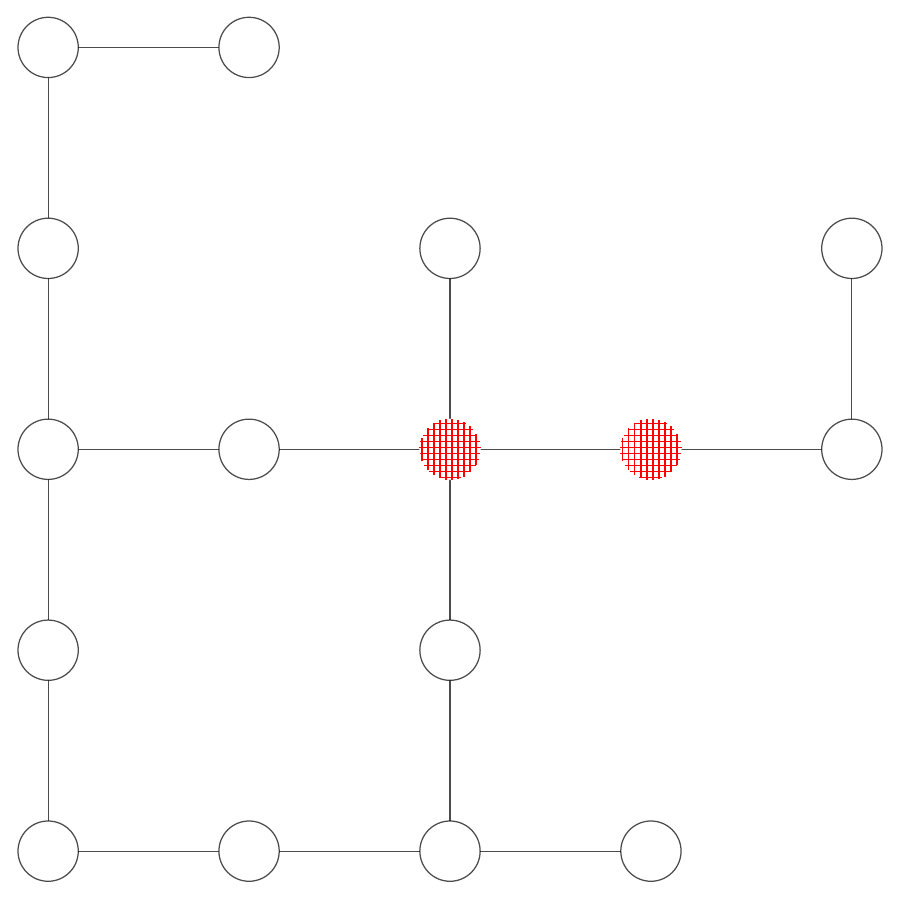}
    }
\caption{Biconnected components of the value of $\upperBoundExt_2$ shown in Figure \ref{fig:articulation_vertices_upper}.}
\end{figure}

Conditioning on these cut vertices all being \up{} leads to the importance resampling algorithm shown given in Algorithm \ref{alg:fixed_splitting_and_conditioning_for_network_reliability}. Note that the conditioning step is a type of importance sampling and leads to an algorithm involving weighted particles. The factor of $p^{\abs{V_{\mathrm{cut}} \setminus \mathbf y}}$ in the final estimator is the importance weight. 

The cut vertices of a graph can be used to decompose it into \emph{biconnected components}. A \emph{biconnected graph} is a graph which has no cut vertices, and therefore cannot be disconnected by the removal of a single vertex. A \emph{biconnected component} is a maximal biconnected subgraph of some underlying graph. The value of $\upperBoundExt_2$ shown in Figure \ref{fig:articulation_vertices_upper} has only one biconnected component with more than two vertices. This subgraph is shown in Figure \ref{fig:biconnected_components_1}, and another component is shown in Figure \ref{fig:biconnected_components_2}. Note that the biconnected components are not disjoint and share vertices. Any vertex contained in more than one biconnected component must be a cut vertex. 

Proposition \ref{prop:articulation_biconnected_decomposition} shows that if $F_{R-1}$ occurs, then $\Vup$ is connected if and only if $V_{\mathrm{cut}}\subseteq \Vup$ and the `restriction' of $\Vup$ to every biconnected component of $\upperBoundExt_{R-1}$ is a connected graph. But the connectivity of $\Vup$ restricted to one component is independent of the connectivity of $\Vup$ restricted to another component. We can use this independence to estimate $\P\(F \lmid \lowerBoundExt_{R-1} = \lowerboundExt_{R-1}\)$ more efficiently. This leads to Algorithm \ref{alg:sis} (SIS). 

\begin{algorithm}[!htb]
\SetAlgoSkip{}
\DontPrintSemicolon
\SetKwData{Weights}{\ensuremath{\mathscr{W}}}\SetKwData{CutVertices}{cut}\SetKwData{BiConnected}{bi}\SetKwData{Count}{count}\SetKwData{Resid}{resid}
\SetKwInOut{Input}{input}\SetKwInOut{Output}{output}
\Input{Connected graph $\fixedGraph$, maximum radius $R$, reliability $p$, integer splitting factors $k_0, \dots, k_{R-1}$, initial number of particles $N$ }
\Output{Estimate of RCR}
Identical to lines 1 - 14 of Algorithm \ref{alg:fixed_splitting_and_conditioning_for_network_reliability}\;
\For{$\mathbf d_{R-1} \in \mathscr{Y}$}
{
  $\CutVertices \leftarrow $ cut vertices of $\mathrm{Up}_{R-1}\(\mathbf d_{R-1}\)$\;
  $\BiConnected \leftarrow $ biconnected components of $\mathrm{Up}_{R-1}\(\mathbf d_{R-1}\)$\;
  $\Count_1\leftarrow 0, \dots, \Count_{\abs{\BiConnected}} \leftarrow 0, \Resid \leftarrow 0$\;
  \For(\tcp*[f]{Simulate most copies efficiently}){$i = 1$ \KwTo $\floor{\sqrt[\abs{\BiConnected}]{k_{R-1}}}$}
  {
    sample $\mathbf D$ from $\(\lowerBoundExt_R \lmid \lowerBoundExt_{R-1} = \mathbf d_{R-1}, \CutVertices \subseteq \Vup\)$\;
    \For{$j = 1$ \KwTo $\abs{\BiConnected}$}
    {
      \lIf{$\varphi\(\mathbf D\cap \BiConnected_j\)$}
      {
        $\Count_j \leftarrow \Count_j+1$
      }
    }
  }
  \For(\tcp*[f]{Simulate remaining copies}){$i = \floor{\sqrt[\abs{\BiConnected}]{k_{R-1}}}^{\abs{\BiConnected}} + 1$ \KwTo $k_{R-1}$}
  {
    sample $\mathbf D$ from $\(\lowerBoundExt_R \lmid \lowerBoundExt_{R-1} = \mathbf d_{R-1}, \CutVertices \subseteq \Vup\)$\;
    \lIf{$\varphi\(\mathbf D\)$}
    {
      $\Resid \leftarrow \Resid + 1$
    }
  }
  add $p^{\abs{\CutVertices\setminus \mathbf d_{R-1}}} \(\prod_{j=1}^{\abs{\BiConnected}}\Count_j + \Resid\)$ to \Weights
}
\KwRet{$\(N\prod_{r=0}^{R-1} k_{r}\)^{-1}\sum_{w \in \Weights}w$}
\caption{Sequential imp. sampling algorithm for RCR (SIS)\label{alg:sis}}
\end{algorithm}

We can also apply the same conditioning ideas at steps other than $R - 1$. This has the additional complication that not every cut vertex of $\upperBound_r$ is in fact required to be \up{} in order for $\Vup$ to be connected. However if $v_{\mathrm{cut}}$ splits $\upperBound_r$ into components, at least two of which contain vertices of $\lowerBound_r$, then $v_{\mathrm{cut}}$ must be \up{}. This property can be checked by performing a depth-first search started at all vertices adjacent to $v_{\mathrm{cut}}$. This much more `aggressive' conditioning leads to sample paths with significantly different weights. It is therefore sensible to substitute resampling for splitting, as this tends to automatically remove sample paths with small weights. These ideas lead to Algorithm \ref{alg:sir} (SIR).

\begin{algorithm}[!htb]
\SetAlgoSkip{}
\DontPrintSemicolon
\SetKwData{Weights}{\ensuremath{\mathscr{W}}}
\SetKwData{AverageWeight}{averageWeight}
\SetKwData{CutVertices}{cut}
\SetKwData{ConditionVertices}{cond}
\SetKwData{Components}{components}
\SetKwData{Comp}{comp}
\SetKwData{BiConnected}{bi}\SetKwData{Count}{count}\SetKwData{Resid}{resid}
\SetKwInOut{Input}{input}\SetKwInOut{Output}{output}
\Input{Connected graph $\fixedGraph$, maximum radius $R$, reliability $p$, initial number of particles $N$ }
\Output{Estimate of RCR}
$\mathscr{Y}\leftarrow \emptyset$ \tcp*{Samples that have hit an intermediate level}
$\Weights\leftarrow \emptyset$ \tcp*{Importance weights}
\For{$i = 1$ \KwTo $N$}
{
  sample $\mathbf D$ from $\lowerBoundExt_0$\;
  \lIf{$I_0\(\mathbf D\)$}
  {
    add $\mathbf D$ to $\mathscr{Y}$, add $1$ to \Weights
  }
}
\For{$r = 1$ \KwTo $R$}
{
  $\AverageWeight \leftarrow \abs{\Weights}^{-1} \sum_{w \in \Weights}$\;
  $\mathscr{Z} \leftarrow $ with repl. sample of size $N$ from $\(\mathscr{Y}, \Weights\)$ \tcp*{Resampling}
  $\Weights \leftarrow \emptyset, \mathscr{Y} \leftarrow \emptyset$\;
  \For{$\mathbf d_{r-1} \in \mathscr{Z}$}
  {
    $\CutVertices \leftarrow $ cut vertices of $\mathrm{Up}_r\(\mathbf d_{r-1}\), \ConditionVertices \leftarrow \emptyset$\;
    \For{$c \in \CutVertices$}
    {
      $\Components \leftarrow $ connected components of $\mathrm{Up}_r\(\mathbf d_{r-1}\) \setminus c$\;
      $n \leftarrow \abs{\lcu \Comp \in \Components \colon \Comp \cap \mathbf d_{r-1} \neq \emptyset\rcu}$\;
      \lIf{$n \geq 2$}
      {
        add $c$ to $\ConditionVertices$
      }
    }
    sample $\mathbf D$ from $\(\lowerBoundExt_{r} \lmid \lowerBoundExt_{r-1} = \mathbf z, \ConditionVertices \subseteq \Vup\)$ \tcp*{Imp. sampling}
    \If{$I_r\(\mathbf D\)$}
    {
      add $\mathbf D$ to $\mathscr{Y}$, add $\AverageWeight\ p^{\abs{\ConditionVertices\setminus \mathbf z}}$ to $\Weights$\;
    }
  }
}
\KwRet{$N^{-1}\sum_{w \in \Weights}w$} \tcp*{Return average of importance weights}
\caption{Sequential imp. resampling algorithm for RCR (SIR)\label{alg:sir}}
\end{algorithm}

\section{The Transfer Matrix Method\label{sec:transfer_matrix_method}}

In order to compare the different Monte Carlo methods in Section \ref{sec:numerical_results} it is useful to use graphs that are reasonably large and for which the RCR is exactly known. Assume that for $0 \leq i \leq \abs{\fixedVertexSet}$ the number of vertex subsets with $i$ vertices that induce a connected subgraph is $c_i$. Then
\begin{align*}
\P\(\varphi\(\Vup\) = 1\) &= \sum_{i=0}^{\abs{\fixedVertexSet}} \P\(\varphi\(\Vup\) = 1 \lmid \abs{\Vup} = i\){{\abs{\fixedVertexSet}} \choose i} p^{i} \(1 - p\)^{\abs{\fixedVertexSet} - i}\\
&=\sum_{i=0}^{\abs{\fixedVertexSet}} \frac{c_i}{{{\abs{\fixedVertexSet}} \choose i}} {{\abs{\fixedVertexSet}} \choose i} p^{i} \(1 - p\)^{\abs{\fixedVertexSet} - i} =\sum_{i=0}^{\abs{\fixedVertexSet}} c_i p^{i} \(1 - p\)^{\abs{\fixedVertexSet} - i}.
\end{align*}
Computation of the values $c_i$ has been shown to be \#P complete \citep{Sutner1991}. However for certain classes of regular graphs the \emph{transfer matrix method} from statistical physics \citep{Klein1986,Kloczkowski1998,Clisby2012} can be used to compute the integers $c_i$ relatively quickly. Once the numbers $c_i$ are known, the RCR can be computed quickly for any value of $p$, to arbitrary accuracy. 

As an illustration of the transfer matrix method assume that we are interested only in computing the \emph{total} number of connected subgraphs of the $5 \times 5$ grid graph. The graph can be considered as five vertical `slices' consisting of five vertices each. As we move from left to right the  `state' of the graph at a particular point consists of the states of those five vertices and information about the paths that connect them which lie to the left. Six of the 52 possible `states' are shown in Figure \ref{fig:transferMatrix}. In the first example state there are three vertices present, which are all connected together by a path lying to the left. In the fourth example state there are three vertices present, and the first and last are connected by a path lying to the left. The middle vertex is not connected to the other two by a path lying to the left. In the third example state all three vertices are connected by a path which lies to the left, although the lower two vertices are necessarily connected as they are adjacent. These 52 states include two `empty' states where no vertices are present; one which occurs before any `state' which contain a vertex, and one which occurs after a `state' which contains a vertex. 

When a connected subgraph is decomposed in this manner there are only a certain number of possibilities for the initial and final states. For example the first state in Figure \ref{fig:transferMatrix} is not possible as an initial state and state 4 as a final state would result in a disconnected subgraph. Further only a limited number of transitions between states are allowed. The number of connected subgraphs can therefore be written as $\mathbf x^T \mathbf B^4 \mathbf y$. Here $\mathbf x$ is a binary vector containing the value $1$ if a state is a possible initial state for a connected subgraph, and $\mathbf y$ is a binary vector containing the value $1$ if a state is a possible final state for a connected graph. The $52 \times 52$ binary matrix $\mathbf B$ contains the value $1$ if a transition between two states is possible for a connected subgraph. This approach can be applied in general to regular planar graphs. By keeping a running count of the number of vertices we can determine the number of connected subgraphs for every total number of vertices. Our implementation made use of the Eigen linear algebra library \citep*{Eigen}, the MPFR numeric library \citep{MPFR} and Boost C++ libraries.
\begin{figure}
\centering
\includegraphics[scale=1.5]{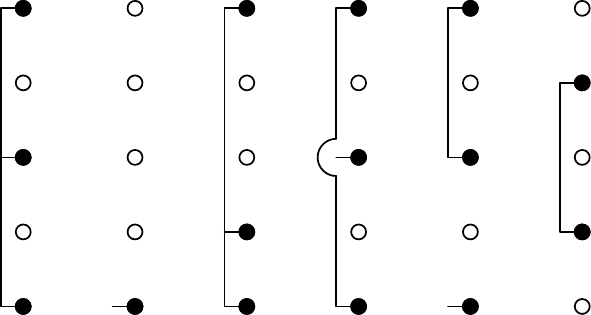}
\caption{A sample of the six of the 52 possible states for the transfer matrix, for the $5 \times 5$ grid graph. \label{fig:transferMatrix}}
\end{figure}

\section{Numerical Experiments\label{sec:numerical_results}}

We performed a simulation study to compare four Monte Carlo (MC) methods for the RCR problem. The four methods were conditional MC \citep{Shah2014}, the RVR method of \cite{Cancela2002}, the sequential importance sampling algorithm given as Algorithm \ref{alg:sis} (SIS) and the sequential importance resampling algorithm given as Algorithm \ref{alg:sir} (SIR). 

We previously assumed that the reliability of every vertex was equal. Not only is this notationally convenient, but it allows us to compute the exact RCR using the transfer matrix method if $\fixedGraph$ is a grid graph. The largest graph for which we applied the transfer matrix method was the $11 \times 11$ grid graph. Although there are $2^{121}$ possible subgraphs, the transfer matrix method allows the computation of the $c_i$ in around an hour. The largest number of connected subgraphs occurred when $i = 75$, for which the number of connected subgraphs was on the order of $10^{31}$. Note that this is well outside the maximum integer value of $2^{64}-1$ representable on a computer without more specialized software. 

We compared the four methods on square grid graphs with length $8, 9, 10$ and $11$. The parameter $p$ was varied between $0.1$ and $0.65$ in steps of $0.05$. At larger and smaller values of $p$ the target event was no longer rare and we do not consider these cases interesting. In addition we also applied all methods to the $14 \times 14$ grid graph, for which it is not possible to compute the unreliability exactly. 

The methods were compared on the basis of their estimated \emph{relative error} (RE) and \emph{work normalized relative variance} (WNRV). For an estimator $\hat{p}$ of a probability $p$ which takes on average time $T$ to compute, these quantities are defined as
\begin{align*}
\mathrm{RE}\(\hat{p}\) &= \frac{\sqrt{\mathrm{Var}\(\hat{p}\)}}{p},&& \mbox{ and } && \mathrm{WNRV}\(\hat{p}\) = \frac{T \mathrm{Var}\(\hat{p}\)}{p^2}.
\end{align*}
In order to estimate the RE and WNRV, each method was run 100 times. Sample sizes for every method were $10^6$ per run. For \SIS the splitting factors are a required input to the algorithm. A separate run with sample size $10^6$ and splitting factors identically $1$ was performed in order to estimate the splitting factors. These estimated values were then used for all the subsequent $100$ runs. The time required to estimate these splitting factors is not included in the work normalized results. 

The simulation study shows that the RVR method is not competitive with the other three methods tested. For the $8 \times 8$ grid graph the average result from the RVR method is numerically accurate for the exactly computed value, although the estimated RE (Figure \ref{fig:re_8}) and WNRV (Figure \ref{fig:wnrv_8}) show that the RVR method is inferior to the other methods. The shading in Figure \ref{fig:re_8} shows 95\% confidence intervals obtained by bias-corrected and accelerated (BCA) bootstrapping \citep{Davison1997} using the R package \texttt{boot} \citep{bootstrapPackage}. In the case of the $11 \times 11$ grid graph the RVR method also performs poorly, with the added problem that for $p$ between $0.2$ and $0.35$ the RVR method is not close to the exact probability (Figure \ref{fig:true_11}). 

\begin{figure}[!tb]
	\centering
	\IfFileExists{./8ggplot.pdf}
		{
			\includegraphics[scale=0.48,page=1]{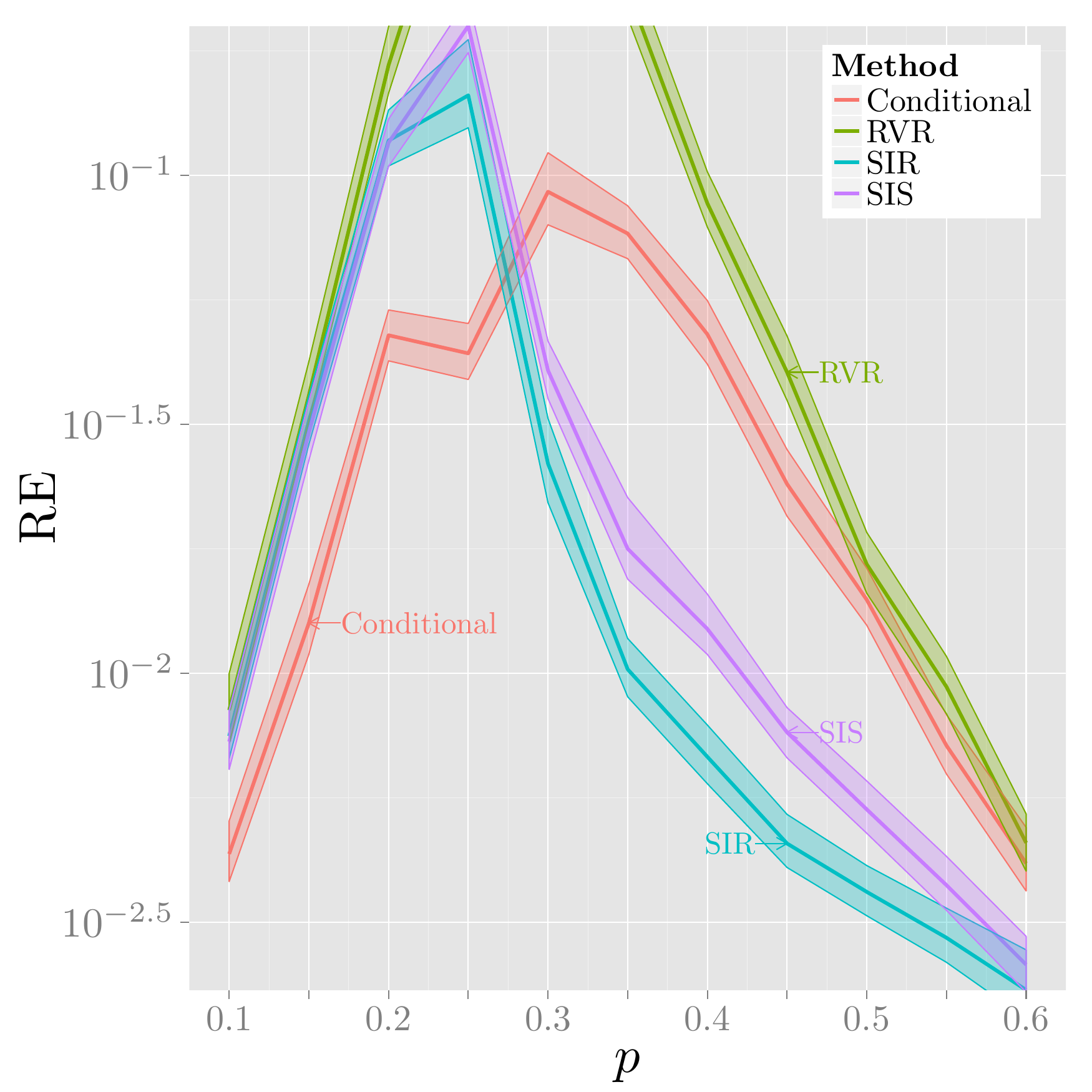}
		}
		{
  			\includegraphics[scale=0.48,page=1]{Simulations/8ggplot.pdf}
		}
	\caption{Relative error results for the $8 \times 8$ grid graph. Shading represents bootstrapped 95\% confidence intervals. \label{fig:re_8}}
\end{figure}
\begin{figure}[!htb]
    \centering
	\IfFileExists{./8ggplot.pdf}
    	{
			\includegraphics[scale=0.48,page=2]{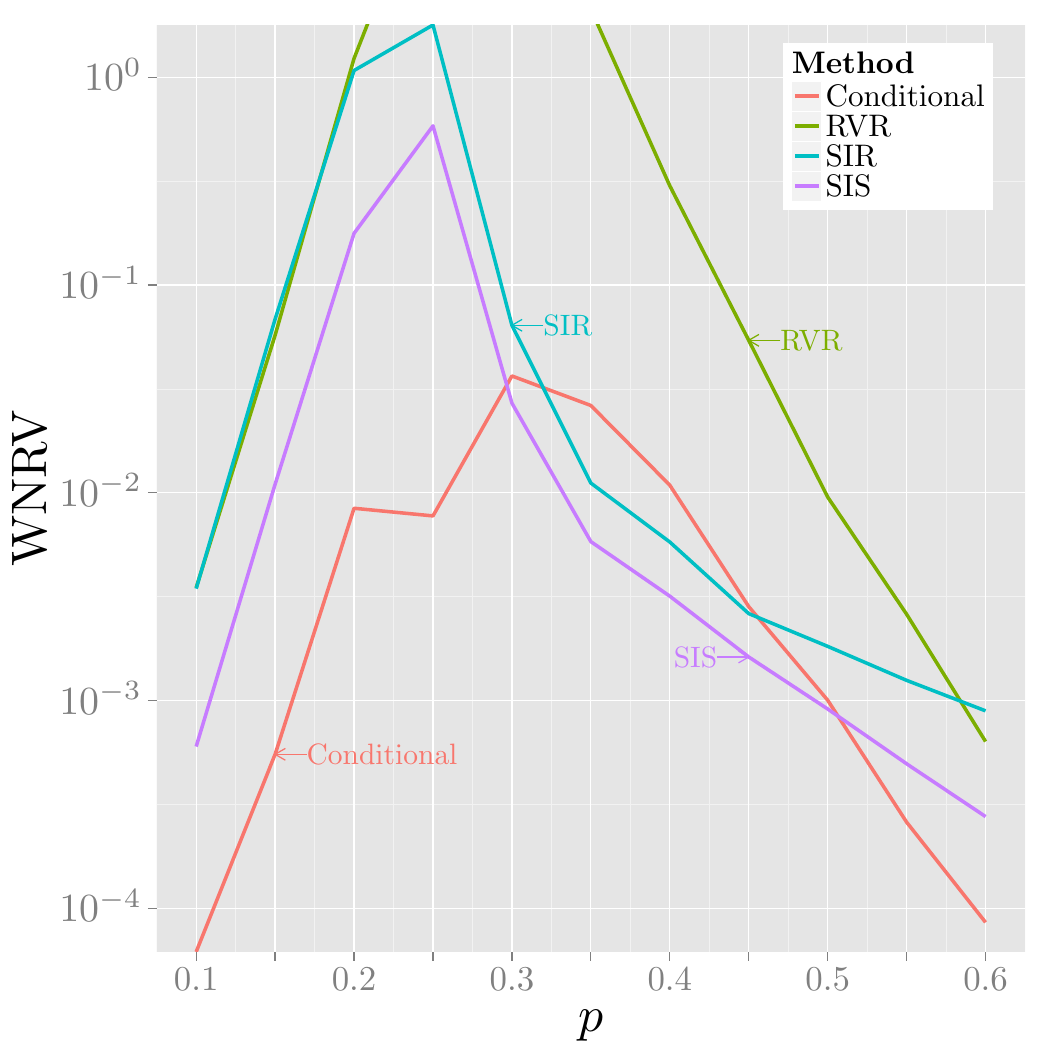}
		}
		{
  			\includegraphics[scale=0.48,page=2]{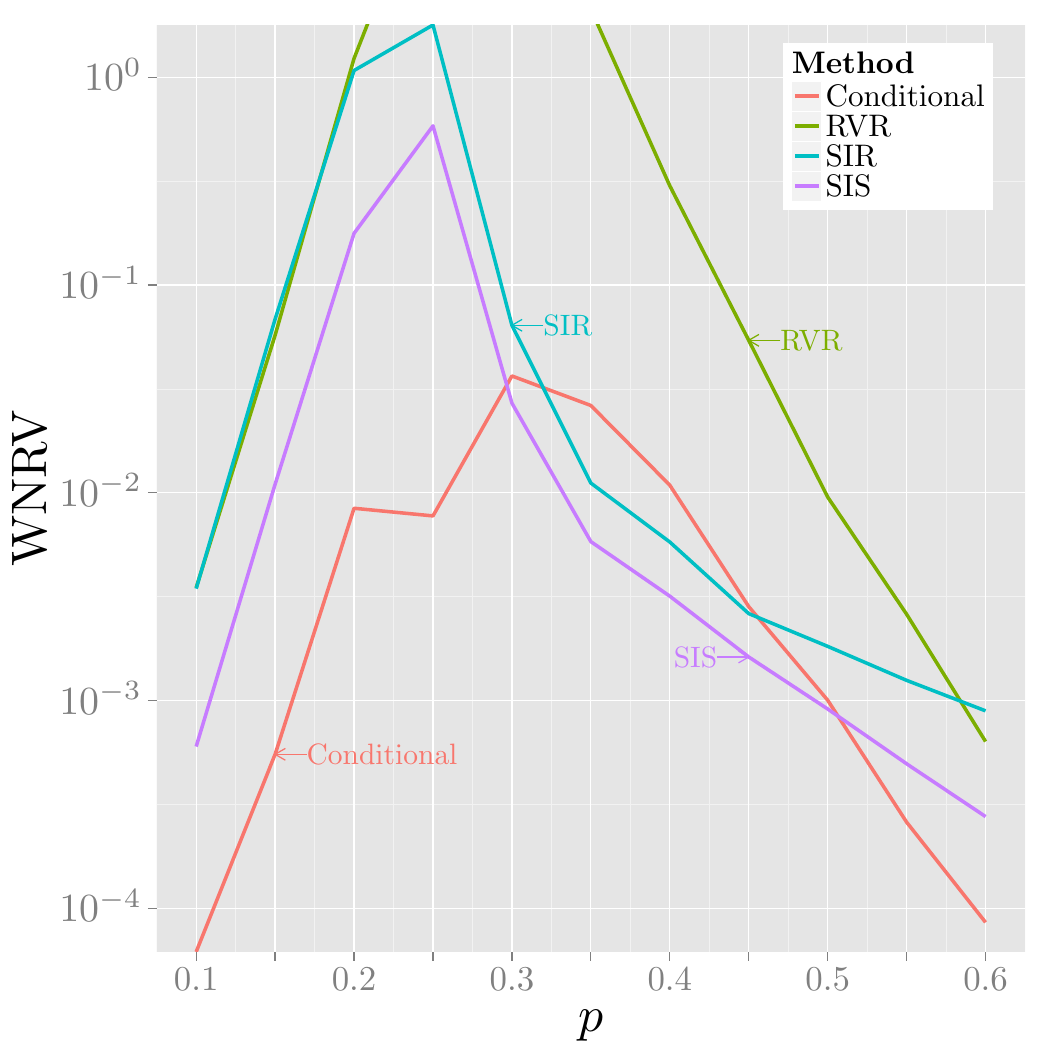}
		}
    \caption{Work normalized relative variance $8 \times 8$ grid graph. \label{fig:wnrv_8}}
\end{figure}
\begin{figure}[!htb]
    \centering
    	\IfFileExists{./11ggplot.pdf}
    	{
			\includegraphics[scale=0.48,page=3]{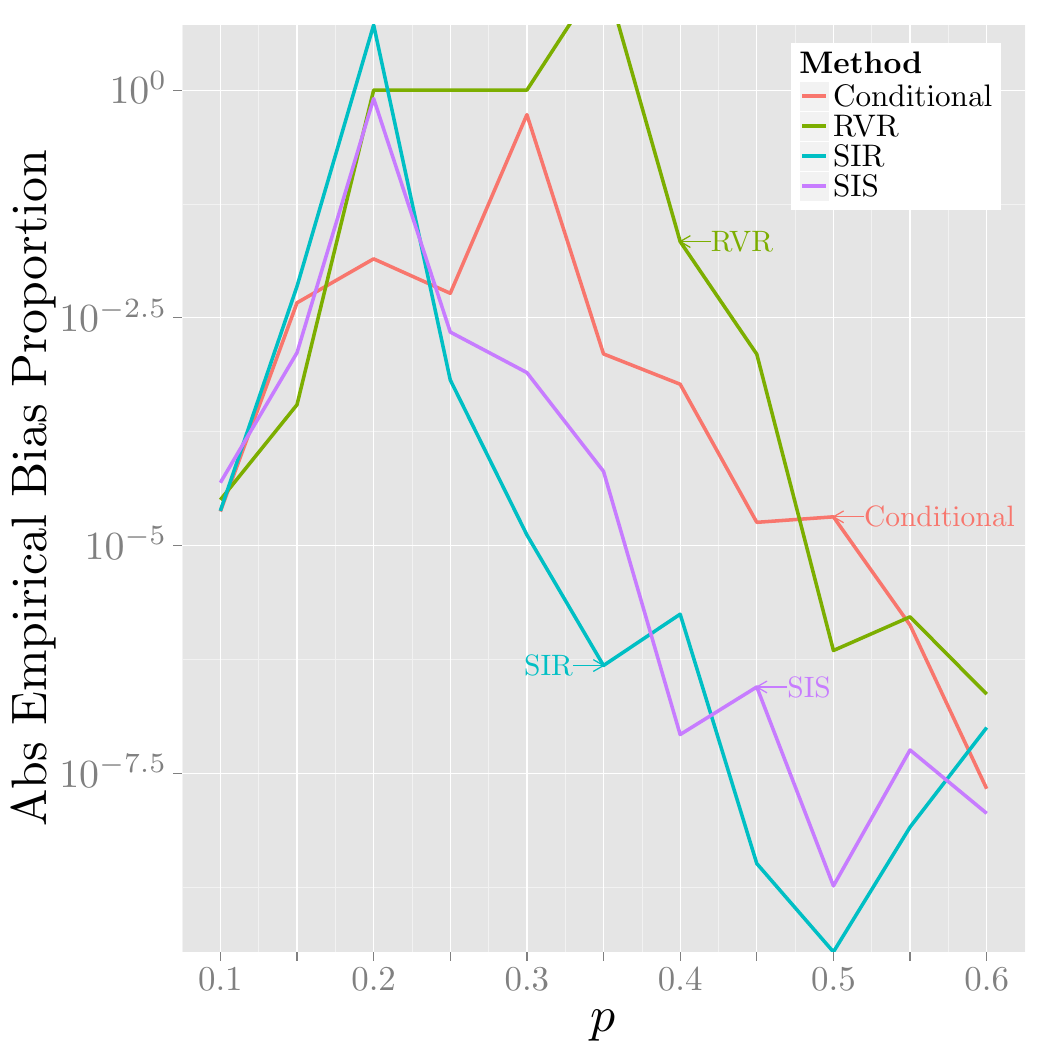}
		}
		{
  			\includegraphics[scale=0.48,page=3]{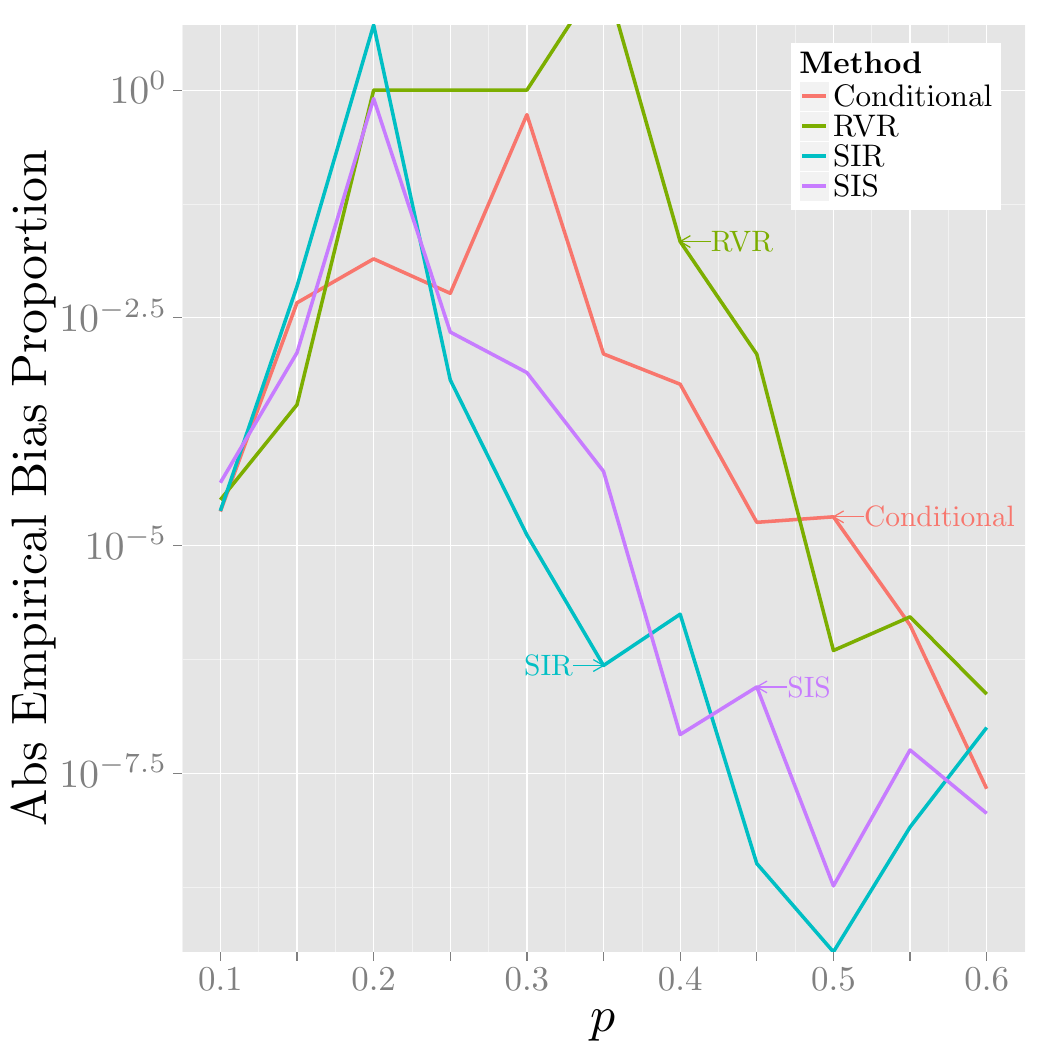}
		}
    \caption{The absolute empirical bias as a proportion of the true value, over 100 replications for the $11 \times 11$ grid graph. \label{fig:true_11}}
\end{figure}

In every case either \SIR or conditional MC has the lowest relative error. Which of these performs best depends on $p$. In general for $p < 0.25$ conditional MC performs best but for $p > 0.25$ \SIR performs best. A similar pattern is observed for the work normalized relative variance, although \SIS and \SIR and conditional MC are somewhat closer in terms of performance. The value of $p = 0.25$ separating these behaviors is believed to be an artifact of our choice of grid values for $p$. The value of $0.25$ is the value that is closest to the value $p^*$ of $p$ that solves
\begin{align*}
p &= \E\[\abs{\Vup}\abs{\fixedVertexSet}^{-1}\lmid F\]. 
\end{align*}
That is, the threshold value $p^*$ is the value of $p$ for which the expected proportion of \up{} vertices under the conditional distribution is exactly $p$. For example, Figure \ref{fig:expected_proportion_11} shows the expected proportion of \up{} vertices for the $11 \times 11$ grid graph, with a vertical line at $p^* = 0.2454$. Note that in Figure \ref{fig:true_11}, for $p < p^*$ \SIS and \SIR do not always converge numerically to the true value using $100$ replications, even on a log scale. Similarly, for $p > p^*$ conditional MC does not always converge to the true value using $100$ replications. 

\begin{figure}[htb]
	\centering
	\includegraphics[scale=0.48,page=1]{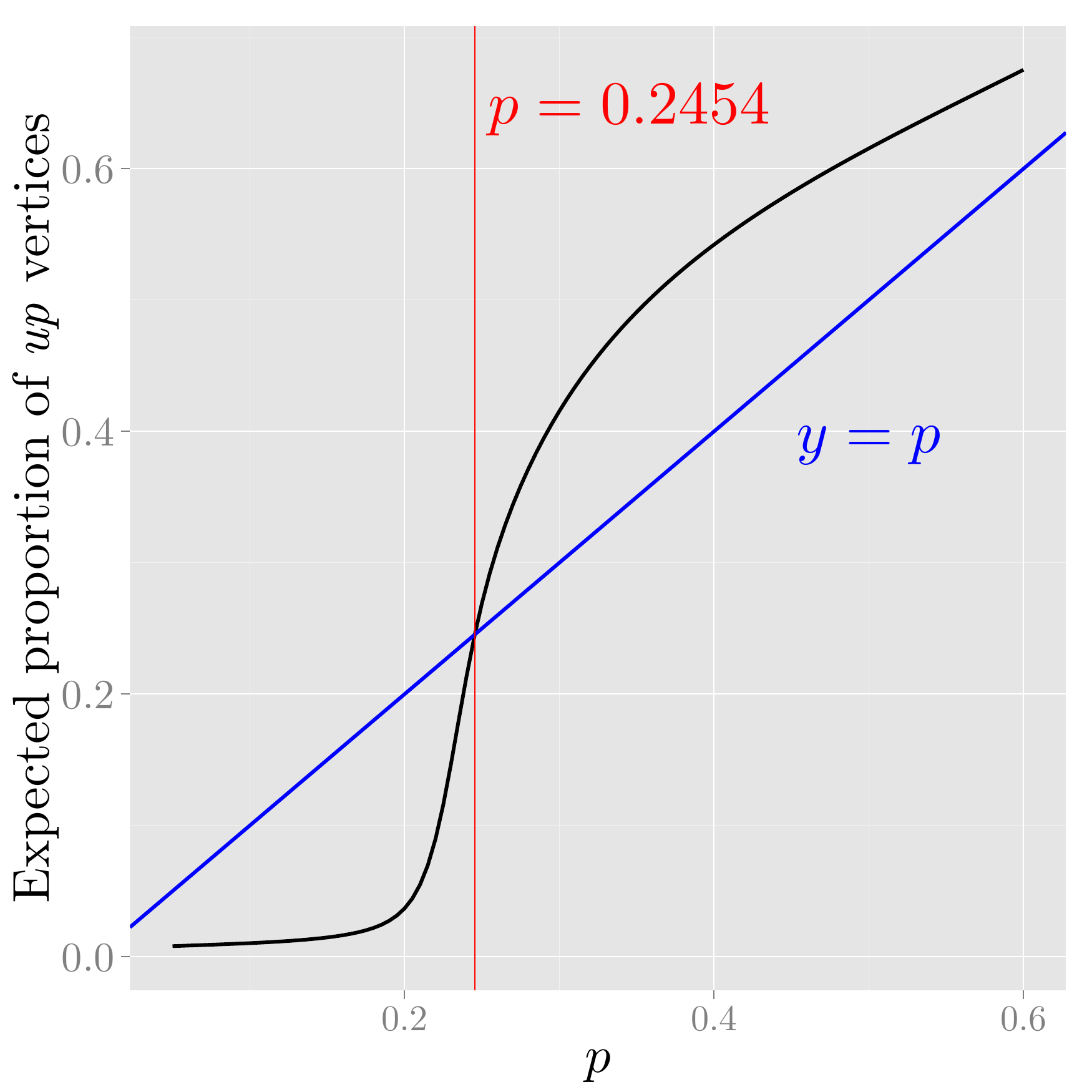}
    \caption{Expected proportion of \up{} vertices for the $11 \times 11$ grid graph, under the conditional distribution. The blue line represents $y = x$, the red line represents $p = 0.2454$.\label{fig:expected_proportion_11}}
\end{figure}

The simulation results suggest that very different strategies must be employed to estimate the RCR, depending on the value of $p$. At some value of $p^*$ (which we assume is unique) the expected fraction of \up{} vertices under the conditional distribution is equal to $p^*$. For values of $p$ slightly smaller than $p^*$ the expected fraction of \up{} vertices is much smaller than $p$, and for values of $p$ slightly larger than $p^*$ the expected fraction of \up{} vertices is much larger than $p$. This occurs because for $p$ smaller than $p^*$ the `easiest' way to obtain a connected graph is to delete all but one of the connected components obtained under the unconditional measure. However if $p$ is large then the easiest way to obtain a connected graph is to add extra vertices, connecting up the components generated under the unconditional measure into a single component. \SIS and \SIR take the approach of adding extra vertices, and are therefore only useful when $p > p^*$. \SIR is more aggressive than \SIS in adding vertices, and the difference in performance between the two cases is therefore more pronounced. Conditional MC can be interpreted as `deleting' vertices, and is therefore efficient for $p < p^*$. 

Figure \ref{fig:re_11} gives the relative errors for the four tested algorithms for the $11 \times 11$ grid graph. The shaded regions indicate 95\% confidence intervals for the RE for the SIR and conditional MC methods, obtained using BCA bootstrapping. The relative errors of \SIR and conditional MC differ by a factor of up to 100, and which method performs better depends on whether $p < p^*$ or $p > p^*$. On a work normalized basis (Figure \ref{fig:wnrv_11}) conditional MC outperforms \SIR by a factor of up to $10^5$ for $p < p^*$, but \SIR outperforms conditional MC by a factor of up to $10^2$ for $p > p^*$. 

\begin{figure}[!htb]
  	\centering
	\IfFileExists{./11ggplot.pdf}
    	{
			\includegraphics[scale=0.48,page=1]{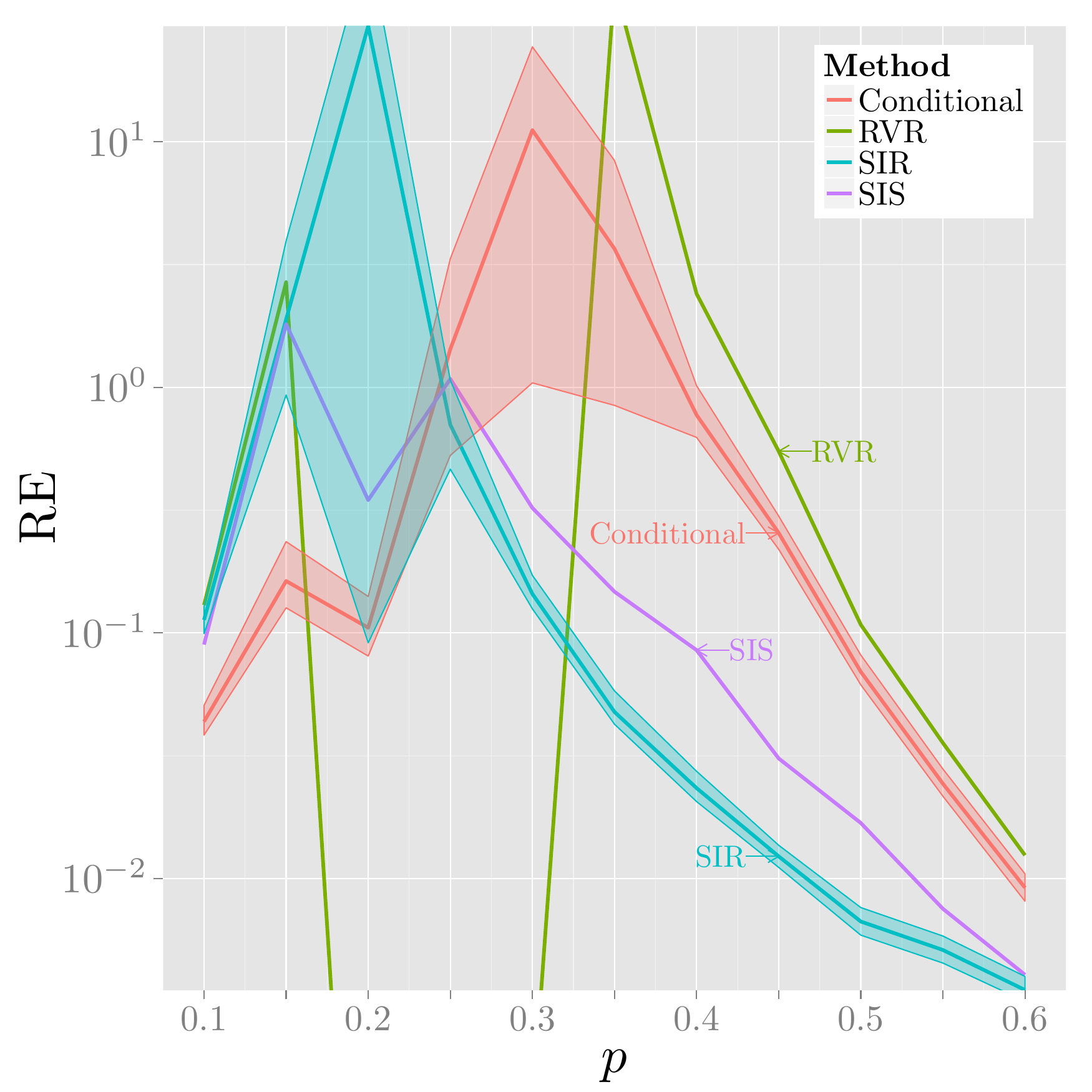}
		}
		{
  			\includegraphics[scale=0.48,page=1]{Simulations/11ggplot.pdf}
		}
	\caption{Relative error results for the $11 \times 11$ grid graph. Shading represents bootstrapped 95\% confidence intervals for Conditional and SIR.\label{fig:re_11}}
\end{figure}
\begin{figure}[!htb]
	\centering
	\IfFileExists{./11ggplot.pdf}
    	{
			\includegraphics[scale=0.48,page=2]{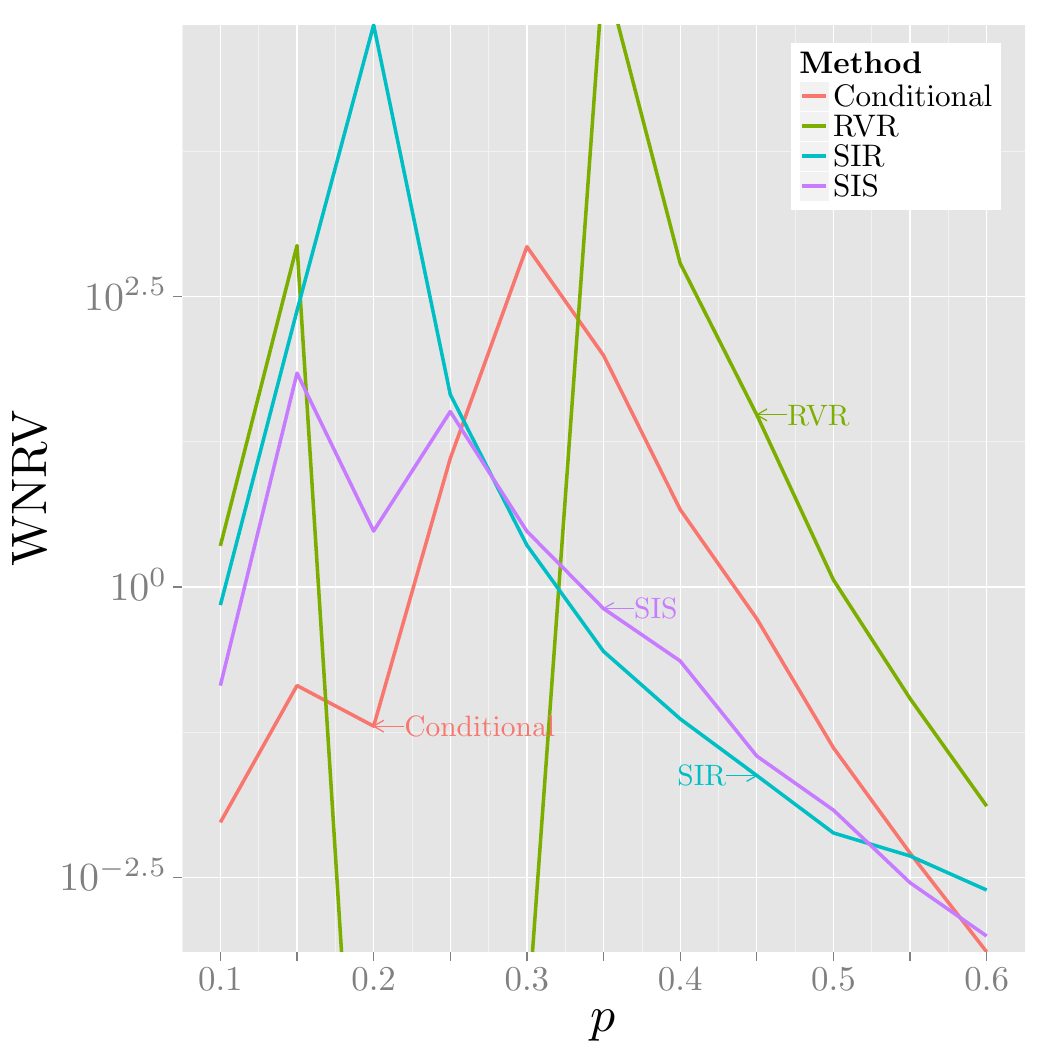}
		}
		{
  			\includegraphics[scale=0.48,page=2]{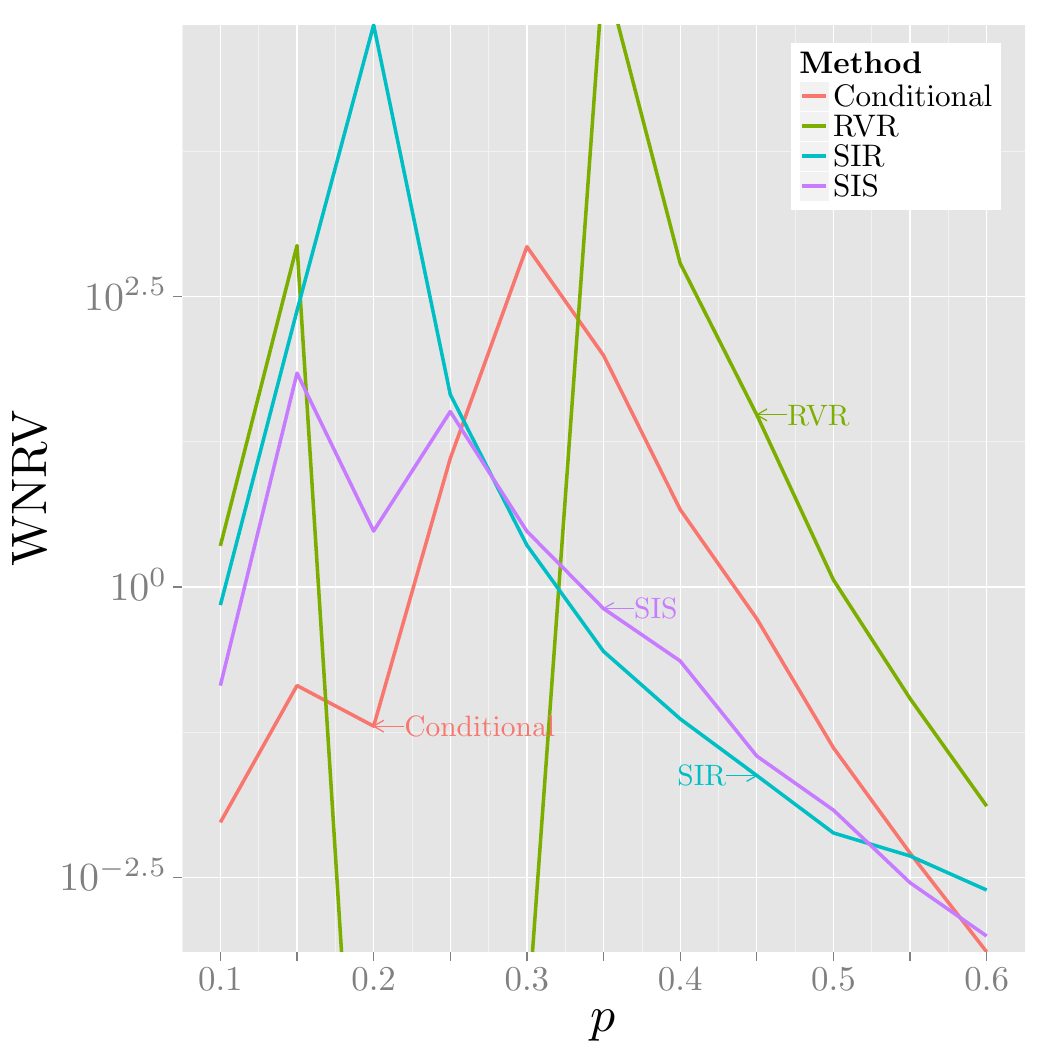}
		}
	\caption{Work normalized relative variance for the $11 \times 11$ grid graph. \label{fig:wnrv_11}}
\end{figure}

Figure \ref{fig:re_14} shows the relative errors of the tested algorithms for the $14 \times 14$ grid graph. The shaded regions indicate 95\% confidence intervals for the RE for the SIR and conditional MC algorithms, obtained using BCA bootstrapping. We do not show results from the RVR method as the estimated relative errors were orders of magnitude higher than those for the conditional, SIR and SIS algorithms. For the purposes of calculating the relative errors we used the estimates given by conditional MC for $p \leq 0.25$ and the estimates given by the SIR method for $p > 0.25$. 

Note that for $p < 0.25$ the SIR method appears to have lower relative error. However the relative error estimates are believed to be inaccurate in this case, similar to the situation for the RVR method on the $11 \times 11$ grid graph (Figure \ref{fig:re_11}). The true RE for the SIR method is believed to be orders of magnitude larger than the RE for conditional MC for $p < 0.25$. The opposite situation occurs for $p \geq 0.25$; the relative error of conditional MC appears to decrease as $p$ increases to $0.35$, but this is believed to be because the RE is poorly estimated in this case. The true RE for conditional MC is believed to be orders of magnitude higher than the RE for the SIR method, for $0.3 < p < 0.45$. The relative errors are believed to be well estimated for all three methods for $0.45 \leq p \leq 0.6$, and in this region the relative error of the SIR method is a factor of $100$ smaller than the RE for conditional MC. For the cases where the relative errors are poorly estimated, it is believed to be computationally infeasible to run the relevant algorithm enough times to obtain an accurate estimate. 
\begin{figure}[htb]
	\centering
  	\IfFileExists{./14ggplot.pdf}
    	{
			\includegraphics[scale=0.48,page=1]{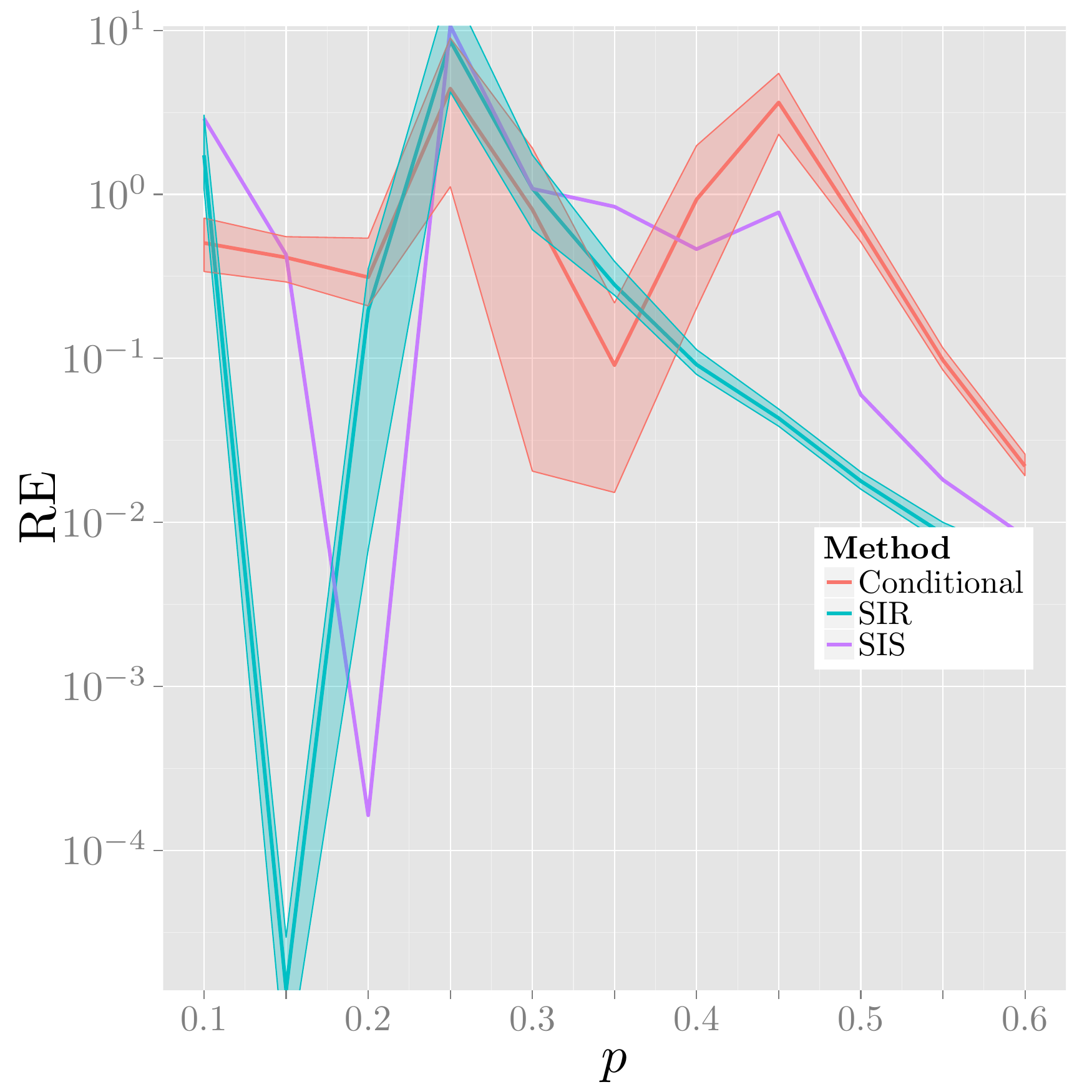}
		}
		{
  			\includegraphics[scale=0.48,page=1]{Simulations/14ggplot.pdf}
		}
	\caption{Relative error results for the $14 \times 14$ grid graph. Shading represents bootstrapped 95\% confidence intervals for the SIR and conditional MC algorithms. \label{fig:re_14}}
\end{figure}

Note that the SIR and SIS algorithms are not well suited to the situation where $\fixedGraph$ is a large subset of a complete graph. In that case every vertex is adjacent to every other vertex. So $\lowerBoundExt_0, \dots, \lowerBoundExt_{R-1}$ represent knowledge of only the first \up{} vertex, and $\lowerBoundExt_R$ represents knowledge of the state of \emph{every} vertex. In this case the SIR and SIS algorithms we describe are similar to crude MC or the RVR method. Further work is required to find methods that behave well in this situation. Another direction for further work is to find algorithms that efficiently estimate the RCR in the regime where $p$ is small. Although conditional MC works acceptably well, it is somewhat unsophisticated. For the special value $p = p^*$ (which appears to be around $0.25$ in our examples) none of the algorithms we tested work well. As shown in Figure \ref{fig:re_11}, the SIS, SIR and conditional MC algorithms all had a relative error close to $1$ in that case. This case appears to be particularly difficult as methods based around increasing or decreasing the number of \up{} vertices will not be effective.

\section{Acknowledgments}

This work was supported by the Australian Research Council Centre of Excellence for Mathematical \& Statistical Frontiers, under grant number CE140100049. 

\appendix 

\section{Proofs}
\label{sec:omitted_proofs}
\begin{sectionprop}\label{prop:upper_bound_proof}
The random variable $\upperBound_r$ is a superset of $\Vup$.
\end{sectionprop}
\begin{proof}
Recall that $\lowerBound_r \subseteq \upperBound_r$. Assume that there is a vertex $x'$ of $\Vup$ which is not in $\upperBound_r$ and therefore not in $\lowerBound_r$. If $x' > \max \lowerBound_r$, then we must have $d\(x', \lowerBound_r\) > R - r$ to avoid having $x' \in \upperBound_r$. So $\mathrm{Next}\(\lowerBound_r, r\) = x'$. This implies that $x' \in \lowerBound_r$, which is a contradiction. 

Assume that $x' < \max \lowerBound_r$, and define the set
\begin{align*}
N &= \lcu n \in \mathbb Z^+ \lmid x' > \max \lowerBound[n]_r \rcu. 
\end{align*}
There is a maximum element $l$ of this set, implying that $x' > \max \lowerBound[l]_r$ and $x' < \max \lowerBound[l+1]_r$. As $x'$ is not in $\upperBound_r$ we must have $d\(x', \lowerBound_r\) > R - r$. This means that $\mathrm{Next}\(\lowerBound[l]_r, r\) = x'$, in which case we would have $x' \in \lowerBound_r$, which is a contradiction. 

Therefore every vertex of $\Vup$ must be contained in $\upperBound_r$. 
\end{proof}

\begin{sectionprop}\label{prop:min_determined_by_contained}
Take any $0 \leq r \leq R$ and $\lowerbound_r \in \possibleLower_r$. Then for any $\valueVertexSet$ so that $\lowerbound_r \subseteq \valueVertexSet \subseteq \mathrm{Up}_r\(\lowerbound_r\)$, it holds that $\min \lowerbound_r = \min \valueVertexSet$. 
\end{sectionprop}
\begin{proof}
Clearly $\min \lowerbound_r = \min \mathrm{Up}_r\(\lowerbound_r\)$ by the definition of $\mathrm{Up}_r$ in (\ref{eq:defn_up_function}). But $\lowerbound_r \subseteq \valueVertexSet \subseteq \mathrm{Up}_r\(\lowerbound_r\)$ implies that
\begin{align*}
\min \lowerbound_r \geq \min \valueVertexSet \geq \min \mathrm{Up}_r\(\lowerbound_r\). 
\end{align*}
So $\min \lowerbound_r = \min \valueVertexSet$. 
\end{proof}

\begin{sectionprop}Assume that $\valueVertexSet = \lcu x_1,\dots, x_n\rcu$ is an ordered sequence of vertices so that $d\(x_k, \lcu x_1, \dots, x_{k-1}\rcu\) > R - r$. Then $\possibleLower_r$ consists exactly of all such sequences $\valueVertexSet$.
\end{sectionprop}
\begin{proof}
It is clear that $\generateSubset{\valueVertexSet, R, r} = \valueVertexSet$ for any such $\valueVertexSet$. But the definition of $\generateSubsetSymbol$ in Algorithm \ref{alg:subset_generation} will always generate vertices with the stated property. 
\end{proof}

\begin{sectioncorollary}\label{cor:removing_points}If $\valueVertexSet \in \possibleLower_r$, then the value obtained by removing any number of points is still in $\possibleLower_r$. 
\end{sectioncorollary}

\begin{sectionprop}\label{prop:contained_between_determines_lower_bound}
Take any $0 \leq r \leq R$ and $\lowerbound_r \in \possibleLower_r$. Then for any $\valueVertexSet$ so that $\lowerbound_r \subseteq \valueVertexSet \subseteq \mathrm{Up}_r\(\lowerbound_r\)$, 
\begin{align*}
\generateSubset{\valueVertexSet}{R}{r} &= \lowerbound_r. 
\end{align*}
\end{sectionprop}
\begin{proof}
Let $\mathbf y = \generateSubset{\valueVertexSet}{R}{r}$ and let $\valueVertex_r, \dots, \valueVertex^{n_x}$ be an enumeration of $\lowerbound_r$ according to the arbitrary ordering. Similarly let $y^1, \dots, y^{n_y}$ be an enumeration of the set $\mathbf y$. From Proposition \ref{prop:min_determined_by_contained}, we know that $y^1 = \valueVertex_r$. All the subsequent points in $\mathbf y$ are more than distance $R - r$ from $y^1 = \lowerbound[1]_r$. Let $A = \(y^1, \infty\) = \(x^1, \infty\)$. Then 
\begin{align*}
\mathbf y \cap A \subseteq \mathbf x \subseteq \mathrm{Up}_r\(\mathbf y \cap A\), 
\lowerBound_r \cap A \subseteq \mathbf x \subseteq \mathrm{Up}_r\(\lowerBound_r \cap A\).
\end{align*}
Taking the intersection with $A$ removes only the first point of both $\mathbf y$ and $\lowerBound_r$, so by Corollary \ref{cor:removing_points} these are still in $\possibleLower_r$. Applying Proposition \ref{prop:min_determined_by_contained} again shows that $y^2 = x^2$. Applying this argument inductively shows that $n_y = n_x$ and $\mathbf y = \lowerbound_r$. 
\end{proof}

\begin{sectionprop}\label{prop:articulation_vertices_of_second_last}
Assume that $\lowerboundExt_{R-1}$ is a value of $\lowerBoundExt_{R-1}$ which implies that $F_{R-1}$ occurs and let $\upperboundExt_{R-1}$ be the associated value for $\upperBoundExt_{R-1}$. If $\fixedVertex$ is a cut vertex of $\upperboundExt_{R-1}$ then
\begin{align*}
F \cap \lcu \lowerBoundExt_{R-1} = \lowerboundExt_{R-1} \rcu \subseteq \lcu \lowerBoundExt_{R-1} = \lowerboundExt_{R-1} \rcu \cap \lcu \fixedVertex \in \Vup\rcu. 
\end{align*}
\end{sectionprop}
\begin{proof}
If $\fixedVertex \in \lowerBoundExt_{R-1}$ the statement is trivial as $\lowerBoundExt_{R-1} \subseteq \Vup$. So assume that $\fixedVertex \not\in \lowerBoundExt_{R-1}$. 

As $\fixedVertex$ is a cut vertex, it separates $\upperboundExt_{R-1}$ into two non-empty components, denoted by $C_1$ and $C_2$. Assume that component $C_1$ does not contain any vertices of $\lowerboundExt_{R-1}$ and let $\customVertex{w}$ be a vertex in $C_1$. By the definition of $\upperboundExt_{R-1}$ we must have $\customVertex{w} \in \upperbound_{R-1}$. This implies that there is a vertex of $\lowerbound_{R-1}$ contained in $B\(\customVertex{w}, 1\) \subseteq C_1 \cup \fixedVertex$. As we assumed that $\fixedVertex \not\in \lowerbound_{R-1}$, we must have a vertex of $\lowerbound_{R-1}$ contained in $C_1$. 

This implies that $C_1$ contains a vertex of $\lowerboundExt_{R-1}$, and by an identical argument so does $C_2$. So there are vertices of $\Vup$ in both $C_1$ and $C_2$, and if $\fixedVertex \not\in\Vup$ these vertices will lie in different components of $\Vup$. 
\end{proof}

\begin{sectionprop}\label{prop:articulation_biconnected_decomposition}
Assume that $\lowerboundExt_{R-1}$ is a value of $\lowerBoundExt_{R-1}$ for which $F_{R-1}$ occurs and let $\upperboundExt_{R-1}$ be the associated value for $\upperBoundExt_{R-1}$. Let $V_{\mathrm{cut}}$ be the set of cut vertices of $\upperboundExt_{R-1}$, and let $\mathbf B = \lcu B_i \rcu_{i=1}^{\abs{\mathbf B}}$ be the set of all biconnected components of $\upperboundExt_{R-1}$. Define the events
\begin{align*}
A &= \lcu \lowerBoundExt_{R-1} = \lowerboundExt_{R-1} \rcu,&& B = \lcu V_{\mathrm{cut}} \subseteq \Vup \rcu, &&C = \bigcap_{i=1}^{\abs{\mathbf B}} \lcu \varphi\(\Vup \cap B_i\) = 1\rcu. 
\end{align*}
The notation $\Vup \cap B_i$ refers to the subgraph of $\Vup$ induced by the vertex set $B_i$. Then $F \cap A = A \cap B \cap C$. 
\end{sectionprop}
\begin{proof}
Let $\mathbf H$ be the graph with vertices $\mathbf B$, and edges between $B_i$ and $B_j$ if $B_i \cap B_j \neq \emptyset$. Note that as $\upperboundExt_{R-1}$ was connected, $\mathbf H$ must be too.

Assume now that events $A, B$ and $C$ all occur. Take any vertices $B_i$ and $B_j$ of $\mathbf H$ which are connected by an edge, and pick arbitrary vertices $\customVertex{v}_i \in B_i\cap \Vup$ and $\customVertex{v}_j \in B_j\cap \Vup$. Let $\customVertex{c}$ be some vertex contained in $B_i \cap B_j$. Note that the event $B$ implies that in fact $B_i \cap B_j \subseteq \Vup$, as all vertices in the intersection are cut vertices. 

As $B_i\cap \Vup$ is a connected graph there exists a path $P_i$ in $B_i\cap\Vup$ from $\customVertex{v}_i$ to $\customVertex{c}$ and another path $P_j$ in $B_j \cap \Vup$ from $\customVertex{c}$ to $\customVertex{v}_j$. This gives a path in $\Vup$ connecting $\customVertex{v}_i$ and $\customVertex{v}_j$. The connectivity of $\mathbf H$ means this argument can be extended to the case where $B_i$ and $B_j$ are not adjacent in $\mathbf H$. This proves that, given our assumptions, $\Vup$ is connected. We have therefore proved that $A \cap B \cap C \subseteq F \cap A$.

Now assume that $F$ and $A$ occur. Proposition \ref{prop:articulation_vertices_of_second_last} implies that $F \cap A \subseteq A \cap B$, so event $B$ must occur. Now pick any biconnected component $B_i$ and any two vertices $\customVertex{v}_1, \customVertex{v}_2 \in B_i \cap \Vup$. As $\Vup$ is connected there exists a simple path between $\customVertex{v}_1$ and $\customVertex{v}_2$. If there is such a simple path $P$ that leaves $B_i \cap \Vup$ then the subgraph $P \cup B_i$ of $\upperBoundExt_{R-1}$ is a biconnected component strictly larger than $B_i$. This would contradict $B_i$ being maximal, so $P$ must stay within $B_i\cap \Vup$, making $B_i \cap \Vup$ connected. Therefore event $C$ must also occur.

This proves that $F \cap A = A \cap B \cap C$. 
\end{proof}


\bibliographystyle{spbasic}      
\bibliography{../../thesis,../../randomGraphs,../../networkReliability,../../networkReliabilityApplications,../../residualNetworkReliability,../../sequentialMonteCarlo,../../maxFlow,../../computerScience,../../TransferMatrix}   

\end{document}